\documentclass[10pt,conference,compsocconf]{IEEEtran}

\usepackage{color}
\usepackage[show]{notes}
\usepackage{url}
\usepackage{amsmath,amssymb}
\usepackage{epsfig}
\usepackage{algorithm}
\usepackage{multirow}
\usepackage{xspace}
\usepackage{subfigure}
\usepackage{wrapfig}

\usepackage{graphicx}
\usepackage{amssymb}
\usepackage{epstopdf}
\usepackage{listings}

\usepackage{algorithm}
\usepackage[noend]{algorithmic}
\usepackage{tabu}

\definecolor{Gray}{gray}{0.8}

\definecolor{Amit_color}{rgb}{0.2,0.2,0.8} 
\newcommand{\AG}[1]{{#1}}
\newcommand{\LL}[1]{{#1}}
\definecolor{Glenn_color}{rgb}{0,0.6,0} 
\newcommand{\GB}[1]{{#1}}

\renewcommand{\algorithmicrequire}{\textbf{Input:}}
\renewcommand{\algorithmicensure}{\textbf{Output:}}
\renewcommand{\algorithmiccomment}[1]{// #1}

\newtheorem{theorem}{Theorem}

\newtheorem{problem}{Problem}

\DeclareMathOperator*{\argmax}{arg\,max}

\newcommand{\eat}[1]{}

\newcommand{\M}{\mathcal{M}}
\newcommand{\C}{\mathcal{C}}
\newcommand{\E}{\mathcal{E}}
\newcommand{\calP}{\mathcal{P}} 

\newcommand{\B}{\mathcal{B}}

\renewcommand{\S}{\mathcal{S}}
\newcommand{\spc} { \hspace{2pt}}

\newcommand{\followups}{\textsc{Followups}\xspace}

\newcommand{\probname}{\textsc{Proxi}\xspace}
\newcommand{\NextExplanation}{\textsc{NextExplanation}\xspace}
\newcommand{\MineExplanations}{\textsc{MineExplanations}\xspace}

\newcommand{\cov}{\sigma}
\newcommand{\NPhard}{NP-hard\xspace}

\newcommand{\random}{\textsc{Random}\xspace}
\newcommand{\greedy}{\textsc{Greedy}\xspace}
\newcommand{\mostpopular}{\textsc{Most-Popular}\xspace}
\newcommand{\exhaustive}{\textsc{Exhaustive}\xspace}

\newlength{\figwidth} 
\newlength{\figthree} 
\newlength{\figthreeone} 
\newlength{\figfour} 
\newlength{\figfourone} 
\begin{document}



\title{Validating Network Value of Influencers by means of Explanations} 


\author{
\begin{tabular}{cccc}
Glenn S. Bevilacqua$^\dag$ \hspace{2mm} & Shealen Clare$^\ddag$ \hspace{2mm} & Amit Goyal$^\dag$ \hspace{2mm} & Laks V. S. Lakshmanan$^\dag$ \\
\end{tabular}
\\$ $\\
\begin{tabular}{ccc}
University of British Columbia \\
Vancouver, B.C., Canada \\
{\sf $^\dag$\{lennson,goyal,laks\}@cs.ubc.ca, $^\ddag$shealen.clare@gmail.com} 
\end{tabular}
}

\eat{
\numberofauthors{4}

\author{
\alignauthor Glenn Bevilacqua\\
\affaddr{University of British Columbia}\\
\affaddr{Vancouver, BC, Canada}\\
\smallskip
\email{lennson@cs.ubc.ca}
\and
\alignauthor Shealen Clare\\
\affaddr{University of British Columbia}\\
\affaddr{Vancouver, BC, Canada}\\
\smallskip
\email{shealen.clare@gmail.com}
\linebreak
\and
\alignauthor Amit Goyal \\
\affaddr{University of British Columbia}\\
\affaddr{Vancouver, BC, Canada}\\
\smallskip
\email{goyal@cs.ubc.ca}
\and
\alignauthor Laks V. S. Lakshmanan \\
\affaddr{University of British Columbia}\\
\affaddr{Vancouver, BC, Canada}\\
\smallskip
\email{laks@cs.ubc.ca}
}
}

\maketitle \sloppy

\begin{abstract}
Recently, there has been significant interest in social influence
analysis. One of the central problems in this area is the problem of
identifying \emph{influencers}, such that by convincing these users to
perform a certain action (like buying a new product), a large number of
other users get influenced to follow the action. The client of such an
application is essentially a \emph{marketer} who would target these
influencers for marketing a given new product, say by providing
free samples or discounts. It is natural that before committing
resources for targeting an influencer the marketer would be interested
in validating the influence \AG{(or network value)} of influencers returned.
\GB{This requires} digging deeper into such analytical
questions as: who are their followers, on what actions (or products) they 
are influential, etc. However, the current approaches to identifying
influencers largely work as a black box in this respect. The goal of this paper is to
open up the black box, address these questions and provide informative and crisp explanations
\AG{for validating the network value}
of influencers. 

We formulate the problem of providing explanations (called \probname) as
a discrete optimization problem \GB{of feature selection}. We show that \probname is not only
\NPhard to solve exactly, it is \NPhard to approximate within any
reasonable factor. Nevertheless, we show interesting properties of the
objective function and develop an intuitive greedy heuristic. We perform
detailed experimental analysis on two real world datasets -- Twitter and Flixster, 
and show that our approach is useful in generating concise and insightful explanations 
of the influence distribution of users and that our greedy algorithm is
\GB{effective and efficient with respect to} several baselines. 
\end{abstract}

\eat{
The study of social influence analysis has attracted significant
attention ever since Kempe, Kleinberg and Tardos published their seminal
paper on the problem of influence maximization \cite{kempe03}. The
objective is to identify influential users (also called
\emph{influencers}) in a social network such that by convincing these
users to perform a certain action (like buying a new product), a large
number of other users get influenced to follow the action. The client
of such an application is essentially a \emph{marketer} who would target
these influencers for the marketing of the given new product (e.g., by
providing free samples or discounts). 
However, the current approaches to
identifying influencers, and in particular, for the problem of influence
maximization, largely work as black-box and just output the
list of influencers, with an estimate on expected influence spread. The goal of this
paper is to address these questions and provide informative and crisp explanations to
influence distribution of influencers. 
}

\section{Introduction}
\label{sec:intro}

\eat{	
Social influence occurs when one's actions are affected by one's social neighbors. While
the study of social influence has a long history in the fields of
sociology and marketing, it 
} 

The study of social influence has gained tremendous attention in the field of data
mining ever since the seminal paper of Kempe et al.~\cite{kempe03}
on the problem of 
influence maximization. A primary motivating application for influence
maximization as well as for other closely related
problems such as identifying community leaders~\cite{leaders},
trendsetters~\cite{trendsetters}, influential bloggers~\cite{agarwal08}
and microbloggers~\cite{twitterrank,cha10},
is viral marketing.
%
The objective in these problems is to identify influential users (also called
\emph{influencers} or \emph{seeds}) in a
social network 
such that by convincing these users to perform a certain action (like
buying a new product), a large number of other users can be influenced to
follow the action.  The client of such an application is essentially a
\emph{marketer} who would 
target these influencers for marketing a given new product, e.g.,
by providing free samples or discounts. 
It is natural that the marketer would like to analyze the influence
spread \cite{kempe03} or ``network value'' \cite{domingos01} of these influencers 
before actually committing resources
for targeting them. 
She would be interested in the answers to the
following questions: \begin{em}
Where exactly does the influence of an influencer
lie? How is it distributed?
On what type of actions (or products\footnote{buying a product is 
an action.}) is an influencer influential? What are the
demographics of its followers? 
\end{em}   
However, the \emph{current approaches for 
identifying influencers}, and in particular, for selecting seed users for the problem of influence
maximization, \emph{largely work as a black box} \GB{in this respect. Just outputting a}
list of seed users (influencers), along with a scalar which is an estimate of the expected 
influence spread. The goal of this
paper is to open up this black box, address these questions and provide informative and 
crisp explanations for the 
influence distribution of influencers. 


Providing explanations to the marketer for influence spread of influencers can have
several benefits. First, it provides \emph{transparency} to the seed
selection algorithm. The marketer is made aware of why a certain user
is selected as a seed, where the user's influence lies, and on what type
of actions the user is influential. These are important analytics she 
may want to investigate before spending
resources on targeting that seed user. Second, it makes the seed selection algorithm
(and thus the system) \emph{scrutable}. The marketer would now be able
to tell the
system, if the explanations (and thus the algorithm) are wrong or are based
on incorrect assumptions, say using her own surveys or background knowledge.
If the explanations are correct and accurate, \GB{this} can increase the
marketer's \emph{trust} in the system and help her in making
good, informed decisions quicker. In other words,
accurate explanations may enhance \emph{effectiveness} and
\emph{efficiency} of the marketers, in making important marketing
decisions. 
Furthermore, such explanations allow room for 
\emph{flexibility} in the targeting process. Indeed, if the
marketer is not able to target \AG{some of the} seeds successfully, then she
knows what exactly the impact on ``coverage'' would be, and can make adequate
adjustments. For example, if a seed's influence was over young college 
students in CA on handheld devices, the marketer can look for an \GB{alternate} 
seed with close characteristics to cover that demographic. 
Overall, providing accurate and crisp explanations would increase the marketer's 
\emph{satisfaction}, and hence loyalty to the provider of seed selection service, 
\GB{as well as} \emph{confidence} in the seed selection algorithm.
In this paper, we specifically focus on providing explanations for
\emph{influencer validation}, that is, a marketer would be able to
analyze the demographics of the followers and the actions.


\GB{On} the industry side, companies like Klout\footnote{\url{www.klout.com}} and 
Peerindex\footnote{\url{www.peerindex.com}} 
claim to provide users'
influence scores from their Social Media profiles (like Facebook and
Twitter). Both these companies provide ``topics'' over which the
users' influence is spread, in addition to influence scores. The 
explanations we propose offer a principled and comprehensive 
account of just how the influence of a seed user is distributed 
across \GB{the} user (i.e., follower) and action dimensions instead of an 
ad hoc ``influence score'' and ``topic''. 

The merit of providing explanations has been recognized before, in the
related fields of recommender systems (see Chapter 15 of \cite{handbook}
for a survey) and expert systems (see \cite{Lacave2004} for a review).
In these systems, the explanations are known to have benefits similar to
those mentioned above.  {\sl In the field of social influence analysis, to
the best of our knowledge, there has been no such systematic study}.



\begin{table}

	\centering
	\caption{Mike: User with most followups. $k=6$, $l=3$.}
	\label{tab:flixster_firstUser}
	\begin{scriptsize}
	\begin{tabu}{c c c | c c c }   	
		\multicolumn{3}{c|}{} 	&Actions&Followers&Followups\\	
		\multicolumn{3}{c|}{}   & (3.0k) & (106) & (37.9k)\\ \hline
		
		\multirow{2}{*}{rated R} & \multirow{2}{*}{thriller} & male	& 708 & 67 & 6.7k \\ \tabucline[1pt]{3-6}
		& & female & 708 & 37  & 2.3k\\ \tabucline[2pt]{1-6}
		
		\multirow{2}{*}{male}	& \multirow{2}{*}{drama} &len:long& 480 & 67 & 4.8k	\\ \tabucline[1pt]{3-6}
		& &len:med& 520 & 67	& 3.5k\\ \tabucline[2pt]{1-6}
		
		\multirow{2}{*}{pre-1997} & \multirow{2}{*}{comedy} &	male & 607 & 67 &  4.8k\\ \tabucline[1pt]{3-6}
		& & female  & 607 & 37 & 2.2k\\ \hline	
		\multicolumn{5}{r}{Total Coverage:} & 56.3\%\\
	\end{tabu}
	\end{scriptsize}
 
\end{table}

\smallskip
\noindent 
\textbf{A Motivating Example.} 
Consider a social network where users perform actions. For instance, in
Flixster\footnote{See \textsection\ref{sec:exp} for a description of the 
datasets.} (\url{www.flixster.com}), user actions correspond to rating
movies. A movie can have many attributes, e.g., genre, year of release,
length of the movie, rating value, and similarly, users can have many
attributes, e.g., gender, age-group and occupation.
Table~\ref{tab:flixster_firstUser} shows an example of the result from
our \GB{experimental} analysis on the real data of Flixster. In this
example, we analyze the influence of the top influencer in Flixster,
measured in terms of number of followups. For simplicity, we refer to
this user as Mike. Informally, a followup is defined as a follower
following up on a user's action. In the context of this example, the number
of followups of Mike is the number of times a follower of Mike rated a
movie, after he rated it.\footnote{We \GB{exactly define} followup in
\textsection\ref{sec:probdef}}

The table makes the following assertions. Mike has rated 3K movies and
has 106 active followers. In total, he has received 37.9K followups.
This is the overall picture of Mike's influence spread. Drilling down, the 
table shows six explanations describing a partial breakdown 
that covers a significant chunk of the influence mass, each explanation 
corresponding to a row. Each explanation 
is presented in terms of action and user features (first three columns). 
The explanations are heterogeneous: e.g., 
the first two explanations involve attributes maturity rating, genre, and user gender, 
whereas the next two involve the attributes user gender, genre, and movie length. 
\eat{The influence is
described by 6 explanations (one explanation corresponds to one row),
and each explanation contains three features.} 

With just 6 crisp explanations, our algorithm is able to explain 56.3\%
of Mike's followups. As an example, 6.7K followups came from male
followers, on Mike's ratings on thriller  movies rated R. (movies
restricted to persons of age 17 or older).  Moreover, 2.3K followups
came from females on the same category of movies, suggesting that Mike
is quite influential on R-rated thriller movies. The Actions column
tells us that there are  708 such movies (thriller, rated R) that are
rated by Mike, while the Followers column tells us that out of 106
followers, 67 are males and 37 are females and 2 (= 106 - (67 + 37))
others did  not specify their gender. Other explanations reveal that
Mike is also influential on drama movies on males, and old (pre-1997)
comedy movies regardless of gender.

Notice, the entries in each of the numeric columns (Action/Followers Count
and Followups) do \emph{not} sum to their overall values. That is, in
general, the explanations provided do not necessarily completely cover
the entire influence spread of the seed user.  Each explanation consists
of a description involving follower (user) demographics and action
attributes (e.g., topic). It also comes equipped with three statistics
-- action count, follower count, and followups, with the meaning described 
above. It is possible that different explanations cover overlapping
demographics, e.g., females in Vancouver and young college students in
BC. 

There are several benefits to this style of explanations in terms of
affording simple inferences. First, we can deduce that if Mike (in the
example above) rates an arbitrary movie again, he is likely to
receive 37.9K/3K = 12.6 followups on average. Next,
we know on what kinds of movies Mike is influential. Thus, if a
marketer wants to advertise a horror movie, then perhaps Mike is not a
good seed, even though his influence is quite high. Moreover, if Mike
rates a thriller movie rated R, it is likely that it will attract
6.7K/708 = 9.5 followups, from male users on average. 
Clearly, these types
of explanations are very informative and valuable.

Notice that a trivial answer to providing explanations is to describe
every single followup for a given influencer. This is undesirable since
such an explanation would be verbose and uninformative to the marketer
trying to make sense of the influencer's influence. We thus argue for
providing crisp or concise explanations that explain as much of the
influence spread (in terms of followups) of the influencer, as possible.  


\smallskip
\noindent
In the literature, the \emph{network value} of a user \cite{domingos01, kempe03} is treated as 
a scalar, i.e., it's equated with the (expected) influence spread of the user. We
argue that in order to answer the above questions, we must revisit this
notion. Our thesis is that there is much more to the network value of a
user than just a number: it can be seen as a summary of the influence
distribution of the user, which describes how the influence is
distributed, over what kind of user demographics and on what type of
actions. In this paper, we formulate and attack the problem of how to
characterize the distribution of influence of a given seed user. In
particular, we make the following contributions. 


\begin{itemize}
	\item We propose a novel problem of PROviding eXplanations for
		Influencers' validation (\probname) to describe network value of
		a given influencer. We outline several benefits of providing
		explanations. 
	\item We show that \probname is not only \NPhard to solve exactly,
		it is \NPhard to approximate within any reasonable factor.
		However, by exploiting properties of the objective function, we
		develop an intuitive greedy heuristic.
	\item We perform experimental analysis on two real datasets --
		Flixster and Twitter, by exploring the influence distributions
		of top influencers, from both qualitative and quantitative
		angles. 
	\item Performing qualitative analysis, we establish the validity 
		of our framework, while gaining insights into influence
		spread of influencers. On the other hand, with quantitative
		analysis, we show that our algorithm explains a significant amount
		of the influence spread with a small number of crisp
		explanations. We compare our algorithm with various baselines,
		and show that it is both effective and efficient.
\end{itemize}

The rest of the paper is organized as follows. Related work is discussed
in \textsection\ref{sec:related}. We formalize the problem in
\textsection\ref{sec:probdef} and develop our algorithm in
\textsection\ref{sec:algo}. The experimental analysis is presented in
\textsection\ref{sec:exp}, while 
\textsection\ref{sec:concl} summarizes the paper and discusses future work.



\section{Related Work}
\label{sec:related}

We summarize related work under three headings. \\ 
\noindent 
{\bf Identifying Influencers.} Identifying influencers has been extensively studied as the problem of
influence maximization. The first work of this kind is due to 
Domingos et al.~\cite{domingos01}. They refer to users' influence as \emph{network
value} and model it as the expected lift in profit due to influence
propagation. Thus, the network value of a customer is captured as a number.  
Later, Kempe
et al.~\cite{kempe03} formulated this as a discrete optimization
problem: select $k$ influencers in a given social network such that by
targeting them, the expected spread of the influence is maximized, assuming 
the propagation follows a diffusion model such as independent cascades or linear 
threshold or their variants. The
problem is \NPhard. However, the objective function satisfies the nice
properties of monotonicity and submodularity, under the diffusion models considered, 
allowing a simple greedy algorithm to provide a $(1-1/e-\epsilon)$-approximation 
to the optimal solution, for any $\epsilon > 0$
\cite{submodular}. 
\AG{Further exploiting these properties, Leskovec et al.~\cite{LeskovecKDD07} proposed a lazy forward optimization that dramatically improves the
efficiency of the greedy algorithm. The idea is that the marginal gain
of a node in the current iteration cannot be better than its
marginal gain in the previous iterations.}
Goyal et al.~\cite{goyal2012} proposed a direct data driven approach to social
influence maximization. They show this alternative approach is both accurate
(in predicting the influence spread) and is scalable, compared to the 
probabilistic approach of Kempe et al.~\cite{kempe03}. \AG{Their work also highlights the 
importance of validating the influence prediction and spread.}

Considerable work has been done on analyzing social influence on blogs
and micro-blogs. 
Agarwal et al.~\cite{agarwal08}
investigate the problem of identifying influential bloggers in a
community.  They show the most influential bloggers are not
necessarily the most active. Gruhl et al.~\cite{Gruhl2004} analyze 
information diffusion in blogspace by characterizing the individuals and
the topics of their blog postings. 
In \cite{Gomez-RodriguezLK10}, the
authors look into the problem of inferring networks of diffusion and
influence in blogspace.
Weng et al.~\cite{twitterrank} develop a topic sensitive Pagerank-like measure (called Twitterrank) for
ranking users based on their influence on given topics. 
Cha et al.~\cite{cha10} compare three different measures of influence
-- indegree (number of followers), retweets and user mentions, with regard to 
their ability to characterize influencers. They
observe that users who have a large number of followers are not
necessarily influential in terms of spawning off retweets or mentions. 
Romero et al.~\cite{romero2011} showed that the majority of users act as
passive information consumers and do not forward the content to the
network. 
Bakshy et al.~\cite{bakshy2011} find that the largest cascades tend to
be generated by users who have been influential in the past and who have
a large number of followers. 

The problem of identifying influencers, and indeed influence maximization, 
is fundamentally different from our problem \probname. The objective of
\probname is to allow a human (or marketer) to (independently) validate
a given influencer, by generating human
readable, crisp explanations. The explanations consist of features from
action and user dimensions with relevant statistics and are generated in
a way such that they are able to cover the maximum amout influence, 
in terms of followups.

Since our explanations are built of action and user features, works on 
topic-sensitive influence analysis
\cite{TangSWY09, Liu2010, twitterrank, trendsetters, Barbieri2012} and influence
based community detection \cite{leaders, Barbieri2013} are relevant and 
we survey these next. 

\noindent 
{\bf Topics.} 
Tang et al.~\cite{TangSWY09} introduce the problem of
topic-based social influence analysis. Given a social network and
a topic distribution for each user, the problem is to find
topic-specific subnetworks, and topic-specific influence weights
between members of the subnetworks. 
Liu et al.~\cite{Liu2010} propose a generative model which utilizes the
content and link information associated with each node (which can
be a user, or a document) in the network, to mine topic-level direct
influence. They use  Gibbs sampling  to estimate the
topic distribution and influence weights.
Weng et al.~\cite{twitterrank}, as described earlier, propose a topic
sensitive Pagerank-like measure to rank users of Twitter. 
In \cite{trendsetters}, the authors define \emph{trendsetters} as the
``early adopters'' who spread the new ideas or trends before they become
popular. They also propose a Pagerank-like measure to identify
trendsetters. Barbieri et al.~\cite{Barbieri2012} extend classical
propagation models like linear threshold and independent cascade
\cite{kempe03} to handle topic-awareness. Our problem is given a network, 
past information cascades in the form of an action log, 
and a seed node, we need to 
generate a compact explanation of the way the influence spread of 
the seed is distributed, which is not addressed by any of these works.

\noindent 
{\bf Communities.} 
Another related line of work is influence-based community detection
\cite{leaders, Barbieri2013}. Goyal et al.~\cite{leaders} define the
notion of ``tribe-leaders'' -- leaders (or influencers) who are followed
up by the same set of users, on several actions. They apply a pattern
mining framework to discover them. Barbieri et
al.~\cite{Barbieri2013} propose a generative model to detect
communities incorporating  information cascades.

\GB{In} contrast \GB{to} the above mentioned papers, our goal is not to model 
topics or to detect communities, but to \emph{describe the influence
distribution of a given user}, by generating 
explanations consisting of interesting features from action
and user dimensions. 
\emph{To the best of our knowledge, this is the first research study to
provide explanations for the purpose of influencer validation.}






\section{Problem Definition}
\label{sec:probdef}

\GB{We consider a directed social graph, $G = (V,D)$ over a set of users $V$ 
where each arc $(u,v)\in D$ 
indicates that user $v$ follows user $u$,\footnote{Our ideas and 
algorithms easily extend to undirected graphs such as those corresponding 
to friendship links.} 
and a propagation log $\mathbb{L}$, 
a set of triples $(u, a, t_u)$ signifying that user $u$ performed action $a$ at time $t_u$.}
When the action $a$ is clear from the context, by $t_u$ we mean the time at which 
user $u$ performed action $a$. 
We say an action $a$ is \emph{propagated} from
$u$ to $v$ if $(u,v) \in D$, and the log $\mathbb{L}$ cotains the tuples 
$(u,a,t_u)$ and $(v,a,t_v)$ for some $t_u$ and $t_v$, such that $t_u < t_v$. 
This defines a \emph{propagation graph} of $a$ as a directed graph $G(a)
= (V(a), D(a))$, with $V(a) = \{u \in V | \exists t_u: (u, a, t_u) \in
\mathbb{L}\}$ and $D(a) = \{(u,v) \in D | u \in V(a), v \in V(a), \mbox{
and } t_u < t_v\}$.  Define an \emph{influence cube} $\C$ over the
dimensions Users (as influencers), Actions and Users (as followers) as
follows: for a cell $(u, a, v)$, $\C(u,a,v) = 1$ if there exists a
(directed) path from $u$ to $v$ in $G(a)$, i.e., $v$ performed action
$a$ after $u$ did. All other cells have value $0$. 

Given a user $u$, by a \emph{followup} of u, we mean a cell $(u,a,v)$
for which $\C(u,a,v) = 1$.  The \emph{followup set} of $u$ is then the
set of followups of $u$: $\M_u = \{(u,a,v) \mid \C(u,a,v) = 1\}$. When
the user is understood from the context, we use $\M$ instead of $\M_u$.
We assume users are equipped with a set of features (e.g., age, location
etc), and similarly for actions (e.g., topic). Descriptions for followup
sets are derived from attributes by means of predicates of the form
$\texttt{A} = \texttt{val}$ where \texttt{A} is an attribute and 
\texttt{val} is a value from its domain. \LL{We assume numeric attributes are 
binned into appropriate intervals. Thus it suffices to consider 
only equality. 
E.g., year = pre-1997, maturity-rating = ``rated R'', and 
gender = female are
predicates/features.} {\em We use the terms predicates and features
interchangeably}.  Let $\calP$ be the set of all predicates. Consider a
cell $(u,a,v)$ in $\M_u$, the followup set of user $u$, and a predicate
$p \in \calP$, we say the cell satisfies the
predicate, $(u, a, v) \models p$, provided 
\eat{that the $p$ is derived from
the attribute-value of action $a$ or from user $v$. } 
either $p$ is a user predicate and user $u$ satisfies this predicate or 
$p$ is an action predicate  and action $a$ satisfies this predicate. 
%
For a predicate $p$, we define $\M^p
= \{(u,a,v) \mid (u,a,v)\in\M \,\&\,(u,a,v) \models p\}$, i.e., the subset of followups
satisfying the predicate. We define an \emph{explanation} as a
conjunction of one or more (user and/or action) predicates. 
Given an explanation $E$, we define $\M^E = \bigcap_{p \in E} \M^{p}$.
i.e., the set of followups satisfying
all the predicates in $E$. We define the \emph{coverage} of an
explanation to be $\cov(E) = |\M^E|$, i.e., the number of followups
satisfying $E$.

Our goal is to provide explanations for the followup set of a 
user (candidate influencer). On one hand, we would like each explanation 
to be as informative as possible. On the other, the total size of explanations 
should be concise or crisp so that a human (marketer) can quickly 
make sense of them. At the same time, between them, the explanations should 
cover as much ``influence mass'' as possible. We formalize these intuitions by 
insisting that each explanation should have length $\ge l$ and ask for a set 
of at most $k$ explanations $\E = \{E_1, ..., E_k\}$ such that the number of followups 
covered by these explanations is as large as possible. For a set of explanations 
$\E$, we extend coverage as follows: define $\M^{\E} = \bigcup_{E\in\E} \M^E$ and 
finally, define the \emph{coverage} of a set of explanations as 
$\cov(\E) = |\M^{\E}|$, i.e., the number of followps in $\M$ which satisfy 
at least one explanation in $\E$. That is,
\begin{align}
	\cov(\E) = |\M^{\E}| = | \bigcup_{E\in\E} \M^E | = |
	\bigcup_{E\in\E} \bigcap_{p \in E} \M^{p} |
\end{align}

\LL{Note the term coverage is defined for a single explanation as well as for a 
set of explanations. In discussing  
the properties of the coverage function,} \GB{we 
consider }
\LL{both $\sigma(E): 2^{\calP} \rightarrow \mathbb{R}$, coverage of a single explanation as a 
function of the features in the 
explanation, as well as 
$\sigma(\E): 2^{2^{\calP}} \rightarrow \mathbb{R}$, coverage of a set of explanations as a 
function of the explanations 
in the set $\E$. The notation and the context should make it clear.}

The main problem we study in this paper is \probname (PROviding eXplanations for
validating the network value of Influencers): 

\begin{problem}[\probname]\label{prob}
{\em
Given a user $u$, followup set $\M$, the available user and action
predicates $\calP$, and numbers $k$ and $l$, find a set of at most $k$
explanations $\E = \{E_1, ..., E_k\}$, where each explanation $E_i$ is a
conjunction of at least $l$ (user/action) predicates such that
$\cov(\E)$ is maximized.
}
\end{problem}

The lower bound $l$ on the size of each explanation captures the
intuition that explanations should be informative. The upper bound $k$
on the number of explanations captures the intuition that overall the
explanations should be crisp. At the same time, problem asks for the
influence mass covered (coverage) to be maximum.

\subsection{Hardness of PROXI}
Not surprisingly, it turns out that \probname is \NPhard. Unfortunately
though, not only it is \NPhard to solve exactly, it is \NPhard to
approximate within any reasonable factor, even when $k=1$, which in
other words is the problem of generating one explanation
(Thm.~\ref{thm:nphardness}).  We establish
the hardness by exploiting its equivalence with the problem of
\emph{Maximum $l$-Subset Intersection} (MSI for short). However, to develop
intuitions for building our algorithm, we show some interesting
properties of the objective function $\cov(\E)$. In particular, we show
that the function $\cov(\E)$ is monotonically increasing and
submodular (Thm.~\ref{thm:sm1}), while the function $\cov(E)$ is
monotonically decreasing and supermodular (Thm.~\ref{thm:sm2}). We
exploit these results to develop our algorithm
(\textsection\ref{sec:algo}).

\begin{theorem}\label{thm:nphardness}
Problem \probname is \NPhard to solve exactly. Moreover, it cannot be
approximated within a factor of $0.5 n^{1 - 2\epsilon} + O(1 - 3 \epsilon)$ for any
$\epsilon > 0$, unless P=NP.
\end{theorem}

\begin{proof} 
We prove the claim for the special case when $k=1$. In this 
case, the problem reduces to finding exactly one explanation $E$, of
length $l$ such
that $\cov(E)$ is maximized. Since $\cov(E)$
is defined as the intersection of sets $\M^p$ for all $p \in E$, the
problem is equivalent to \emph{Maximum $l$-Subset Intersection} (MSI)
Problem \cite{xavier2012, shieh2012}: Given a collection of sets $\S =
\{S_1, S_2, \ldots, S_m\}$ over a universe of elements $\mathcal{U} =
\{e_1, e_2, \ldots, e_n\}$,
the objective is to select (no less than) $l$ sets $\S' \subseteq \S$,
$\S' = \{S_{i_1}, S_{i_2}, \ldots, S_{i_l}\}$
such that its intersection size, $|S_{i_1} \cap S_{i_2} \cap \cdots S_{i_l}|$, is maximum.

It is easy to see that MSI is equivalent to \probname when $k=1$. An
element $e \in \mathcal{U}$ in MSI corresponds to a cell $(u,a,v)$ with
a value 1 in \probname. Similarly, a set $S$ corresponds to $\M^p$, the
set of cells selected by a predicate $p$. Furthermore, the objective is
equivalent -- select $l$ sets ($l$ predicates) such that the resulting
intersection size is maximum. The formal reduction from MSI to \probname 
and the other way around
is thus straightforward and we skip it for brevity.

It is well known that MSI is \NPhard, and is \NPhard to approximate it within a factor
a $0.5 n^{1 - 2 \epsilon} + O(1 - 3 \epsilon)$ \cite{shieh2012}. Given the equivalence 
between \probname and MSI, the theorem follows.
\end{proof}

\eat{
\note{Laks: ${\cal E}*$ vs. ${\cal E}'$ issue needs to be fixed. Glenn: resolved}
}

\begin{theorem}\label{thm:sm1}
The function $\cov(\E): 2^{2^{\calP}} \rightarrow \mathbb{R}$ is
monotonically increasing and submodular. That is, $\forall \E \subseteq
\E' \subseteq 2^{2^\calP} : \cov(\E) \le \cov(\E')$ and 
$ \cov(\E \cup \{E\}) - \cov(\E) \ge \cov(\E' \cup \{E\}) - \cov(\E')$.
\end{theorem}

\begin{proof}
First, we show that the objective function $\cov(\E)$ is monotonically
increasing. By definition of $\cov(\E)$, we have,
\begin{align*}
\cov(\E \cup \{E\}) - \cov(\E) = |\M^{\E} \cup \M^{E}| - |\M^{\E}|
\end{align*}

Clearly, $\M^{\E} \cup \M^{E}$ is a superset of $\M^\E$, the above
quantity is non-negative, implying that $\cov(\E)$ is monotonically
increasing.
Next, we show the property of submodularity. 
\begin{align*}
\cov(\E' \cup \{E\}) & - \cov(\E') = |\M^{\E'} \cup \M^{E}| - |\M^{\E'}|\\
&= |(\M^{\E'} \cup \M^{E}) \setminus \M^{\E'}| \\ 
&= |(\M^{\E} \cup \M^{\E'} \cup \M^{E}) \setminus (\M^{\E} \cup \M^{\E'}) |  \\
&= |((\M^{\E} \cup \M^{E}) \setminus \M^{\E}) \setminus \M^{\E'} | 
\end{align*}

Since set subraction may only remove set elements,
\begin{align*}
\cov(\E' \cup \{E\})  - \cov(\E') &\le |(\M^{\E} \cup \M^{E}) \setminus \M^{\E} | \\
&\le \cov(\E \cup \{E\}) - \cov(\E)
\end{align*}

which is what we wanted to prove. 
\end{proof}

\eat{
\note{Laks: For monotone non-increasing, isn't it obvious the more features you add 
the smaller the intersection? For submodularity, I wonder whether there is a need for 
$E^*$. Can't we prove it from first principles just using $E$ and $E'$ and $f\not\in E'$? 
The proof looks more complicated than it needs to be. Glenn: resolved} 
}

\begin{theorem}\label{thm:sm2}
The function $\cov(E): 2^\calP \rightarrow \mathbb{R}$ is monotonically decreasing and
supermodular. That is, $\forall E \subseteq
E' \subseteq 2^\calP : \cov(E) \ge \cov(E')$ and 
$ \cov(E \cup \{p\}) - \cov(E) \le \cov(E' \cup \{p\}) - \cov(E')$.
\end{theorem}

\begin{proof}
First, we show that the objective function $\cov(E)$ is monotonically
decreasing. By definition of $\cov(E)$, we have,
\begin{align*}
\cov(E \cup \{p\}) - \cov(E) = |\M^{E} \cap \M^{p}| - |\M^{E}|
\end{align*}

Clearly, $\M^{E} \cap \M^{p}$ is a subset of $\M^E$, the above
quantity is non-positive, implying that $\cov(E)$ is monotonically
decreasing.
We next show the property of supermodularity. 

\begin{align*}
\cov(E' \cup \{p\}) & - \cov(E') = - (|\M^{E'}| - |\M^{E'} \cap \M^{p}|)\\
&= - |\M^{E'} \setminus (\M^{E'} \cap \M^{p})| \\
&= - |(\M^{E} \cap \M^{E'}) \setminus (\M^{E} \cap \M^{E'} \cap \M^{p})|  \\
&= - |\M^{E'} \cap (\M^{E} \setminus (\M^{E} \cap \M^{p}))| 
\end{align*}

Since set intersect may only remove set elements,
\begin{align*}
\cov(E' \cup \{p\})  - \cov(E') &\ge - |(\M^{E} \setminus (\M^{E} \cap
\M^{p}) |\\
&\ge \cov(E \cup \{p\}) - \cov(E)
\end{align*}

which is what we wanted to prove. 
\end{proof}

\section{Algorithm}
\label{sec:algo}

Even though \probname is \NPhard to approximate, the function
$\cov(\E)$ has nice properties as we show in Theorem
\ref{thm:sm1}. Nemhauser et al.~\cite{submodular} show
that maximizing monotonically increasing submodular
functions can be approximated within a factor of $(1-1/e)$ using a
greedy algorithm. Moreover, due to Feige \cite{feige98}, we know that
this is the best possible approximation factor that can be achieved in
polynomial time.  These results, in addition to Theorem \ref{thm:sm1},
suggest that the greedy heuristic which adds the current best explanation
$E$ to $\E$, until $|\E|$ is $k$ would be the best possible heuristic. 
However, the complex step here is to generate one explanation $E$, or
more generally, the next explanation $E$, such that the marginal
coverage $\cov(\E \cup \{E\}) - \cov(E)$ maximized, where $\E$ is the
current set of explanations. We showed that this particular problem is
\NPhard to approximate (see the proof of Theorem \ref{thm:nphardness}).
Thus, strictly speaking, we cannot expect to have an efficient algorithm with a provable
approximation guarantee for \probname.  

However, given the hardness of the problem, we believe that \GB{a} greedy
algorithm of successively generating \GB{explanations by repeatedly picking the best predicate would}
still be a good heuristic. More precisely, in any iteration, where $\E :
|\E| < k$ is the current set of explanations and $E : |E| < l$ is the
current explanation, the greedy algorithm picks the predicate, $p$,
\GB{that when added to $E$ gives an extended explanation that provides the maximum possible additional coverage, w.r.t.~$\E$}. That is,
$\cov(\E \cup \{E \cup \{p\}\}) - \cov(\E)$ is maximum.

Since the search space is massive, a naive greedy algorithm as explained
above would be extremely slow. So we focus our attention on making the
algorithm efficient, by cleverly avoiding unnecessary coverage
evaluations, in any given iteration.  In particular, we optimize our
algorithm by means of lazy evaluation.  Recall, the function $\cov(E):
2^{{\cal P}} \rightarrow \mathbb{R}$ is non-increasing and supermodular.
Thus, the lazy forward approach used by Leskovec et
al.~\cite{LeskovecKDD07} does not work here, as it relies on the
non-decreasing submodular nature of the objective function. We instead
exploit the fact that the coverage of a single explanation $\cov(E): 2^{{\cal
P}} \rightarrow \mathbb{R}$ is non-increasing in the number of features (predicates),
and devise a lazy evaluation optimization based on this.  The idea
is that, while constructing a 
single explanation, the 
\GB{marginal coverage of the explanation} after adding a predicate $p$
to explanation $E$ \GB{also} cannot increase (since $\sigma(E)$ is non-\GB{increasing}). 
Thus, by
maintaining a max-heap of predicates\GB{, $p$,} sorted on additional coverage of 
\GB{the extended explanation $E \cup \{p\}$}
w.r.t.~$\E$, 
we can avoid
coverage recomputations for many of the predicates, in any given iteration.

We next explain our algorithm in detail, given in Algorithms
\ref{alg:MineExplanations} and \ref{alg:NextExplanation}.  In $Q$, we
store the max-heap of features. Each element $p$ in $Q$ represents a
predicate/feature with the following attributes: $p.cells$ denotes the
set of cells corresponding to the predicate, or equivalently $\M^p$;
$p.cov$ denotes the effective additional coverage of \GB{an explanation, $E$,} 
w.r.t.~$\E$, \emph{if} $p$ were 
added to the explanation, that is, $p.cov = \cov(\E \cup \{E \cup
\{p\}\}) - \cov(\E)$.
Due to our lazy evaluation optimization, $p.cov$ may not always store
the correct value. Instead, it may store an outdated value, which may
have been calculated in some earlier iteration. To keep track of it, we
use $p.flag$  to save the iteration when $p.cov$ was last updated.
Moreover, we mark a cell when it is covered by the current set of explanations.
Initially, all cells are unmarked, and $p.flag$ for all the
features is set to $0$. 
The heap $Q$ is sorted on $p.cov$. Our main subroutine \MineExplanations
adds one explanation $E$ at a time, in a greedy fashion, while the
subroutine \NextExplanation generates the next best explanation, again
in a greedy fashion.  




\begin{algorithm}
\caption{\MineExplanations} \label{alg:MineExplanations}
\begin{algorithmic}[1]
\begin{small}
\REQUIRE{$Q$, $k$, $l$} 
\ENSURE{$\E$}
	\STATE $\E \gets \emptyset$.
	\WHILE {$|\E| < k$}
		\STATE $p \gets Q.peek()$.
		\IF {$p.flag < |\E| \cdot l$}
			\STATE $Q.poll()$.
			\STATE $p.cov \gets $ \#cells in $p.cells$ for which
			$cell$ is not marked. 
			\STATE $p.flag \gets |\E| \cdot l$.
			\STATE Reinsert $p$ in $Q$ (and reheapify w.r.t.~$p.cov$).
		\ELSE 
			\STATE $Q' \gets copy(Q)$ (copy includes features'
			coverage).
			\STATE $E \gets \NextExplanation(Q', l, |\E|)$.
			\STATE $\E \gets \E \cup \{E\}$.
		\ENDIF
	\ENDWHILE 
\end{small}
\end{algorithmic}
\end{algorithm}

\begin{algorithm}
\caption{\NextExplanation} \label{alg:NextExplanation}
\begin{algorithmic}[1]
\begin{small}
\REQUIRE{$Q'$, $l$, $|\E|$} \ENSURE{$E$}
	\STATE $E \gets \emptyset$.
	\WHILE {$|E| < l$}
		\STATE $p \gets Q'.poll()$.
		\IF {$p.flag < |\E| \cdot l + |E|$}
			\STATE $p.cov \gets$ \#cells in $p.cells$ for which
			$cell.flag = |\E| \cdot l + |E|$ and $cell$ is not marked. 
			\STATE $p.flag \gets |\E| \cdot l + |E|$.
			\STATE Reinsert $p$ in $Q'$ (and reheapify w.r.t.~$p.cov$).
		\ELSE 
			\STATE $E \gets E \cup \{p\}$.
			\FOR {each $cell \in p.cells$} 
				\STATE \textbf{if} $|E|= l$ \textbf{and} $cell.flag =
				|\E| \cdot l + l$ \textbf{then}
				\STATE \hspace{8pt} mark the \textit{cell}.
				\STATE \textbf{else if} $|E|=1$ \textbf{then} $cell.flag
				\gets|\E| \cdot l + 1$.
				\STATE \textbf{else} $cell.flag \gets cell.flag + 1$.
			\ENDFOR
		\ENDIF
	\ENDWHILE 
\end{small}
\end{algorithmic}
\end{algorithm}

We first describe Algorithm \MineExplanations. 
$\E$ is initialized in line 1 and we iterate until $k$
explanations are generated (lines 2-12). In each of these $k$
iterations, we get the next best explanation by calling \NextExplanation,
and add it to $\E$ (lines 10-12). We make a copy of $Q$ before calling
\NextExplanation, because the ordering of the heap $Q$ can be changed by
the subroutine, and thus corrupt the original ordering. Other lines of
the algorithm implement the lazy evaluation optimization. The feature with
the maximum coverage is taken from the heap $Q$, without removing the
it from $Q$ (i.e., $peek()$, in line 3). 
If there is a
need to recompute the coverage $p.cov$ of the feature $p$ (this
condition is tested using $p.flag$), then we do so in line 6, after
removing it from $Q$ in line 5 (i.e., $poll()$). The flag is
updated in line 7 and the heap $Q$ is re-heapified accordingly.

Next, we describe Algorithm \ref{alg:NextExplanation}, which generates
one explanation at a time. This algorithm also employs a greedy strategy
to select the features in a lazy manner. $E$ is initialized in
line 1 and we iterate until we generate $l$ features in $E$. 
In this subroutine, we also assign a $flag$ to each cell. Intuitively,
$cell.flag$ stores the number of features in the current explanation $E$
that covers the cell. For example, $cell.flag$ is $|\E| \cdot l + 1$,
if it is covered by exactly one feature in $E$, and similarly, it is
$|\E| \cdot l + |E|$ when all features in $E$ cover the cell. It
should be noted that the term $|\E| \cdot l$ is added to ensure that the
flag values across different iterations of Algorithm
\ref{alg:MineExplanations} don't mix up.

We again exploit $p.flag$ to track the iteration when the coverage of
the predicate/feature was last updated. As in $cell.flag$, we also add $|\E| \cdot
l$ in $p.flag$ to avoid the mix up in the flag values across different
iterations.  If there is no need to recompute $p.cov$, then we add $p$
to $E$ (line 9). Next, if $|E| = l$, it implies that this is the last
iteration of Algorithm 2. In that case, we must mark the cells which are
selected by the explanation $E$. Recall that for a cell to be covered by
all features in $E$, its flag should be $|\E| \cdot
l + l$ (lines 11-12). On the other hand, if $|E| = 1$, indicating that
we just selected the first feature in $E$, then $cell.flag$ is initialized to $|\E|
\cdot l + 1$ (line 13). In other cases, that is, $1 < |E| < l$, we
increment $cell.flag$ (line 14). 

Lines 4-7 implement the lazy forward optimization.
If $p.flag$ indicates
that we must recompute $p.cov$, then we do so in line 5. Note that the
coverage of feature $p$ here is the number of cells in $p.cells$ that
are not marked, and which are covered by all the features in $E$. We
test this condition by checking that $cell.flag$ is $|\E| \cdot l +
|E|$. We update $p.flag$ in line 6, and re-heapify $Q'$
accordingly in line 7.

\section{Experiments}
\label{sec:exp}

The goals of our experimental analysis are manifold. Not only we are interested
in identifying influential users, that is, users with high number of followups,
we are also interested in exploring the distribution of their influence, from both
quantitative and qualitative angles. We achieve this by performing an exhaustive
analysis on two real-world datasets -- Flixster and Twitter. We next describe the datasets.

\subsection{Datasets}

\noindent
\textbf{Flixster.}
Flixster (\url{www.flixster.com}) is a major player in the mobile and
social movie rating business. Originally collected by Jamali et
al.~\cite{Jamali10}, the dataset contains 1M users and 7.06M ratings,
distributed across 49K movies. Out of these 1M users, 148K users have provided 
one or more ratings and they have 2.43M edges among them.
Here, an action is a user
rating a movie. 
For each rating, the dataset also includes the
timestamp at which the user rated the movie. 


\textit{User Features:} There are two user attributes in the dataset:
Gender and Age. 
As is done 
in other public datasets such as Movielens, we bin age values  into 7 age
ranges as follows: less than 18, 18--24, 25--34, 35--44, 45--49, 50--55
and 56+. Thus, we ended up with 9 binary user features.

\textit{Action Features:} 
To enrich action (movie) features, we queried IMDB API
(\url{imdbapi.com}) with the movie titles present in the
Flixster data set: 82\% of the movie titles found matches; we ignore
ratings on the remaining 18\% of the unmatched movies, which constituted
9\% of the ratings.
\eat{
\note{Laks: How many or what \% ratings did we lose as a result?} 
}
As shown in Table
\ref{tab:flixster}, 7 attributes were collected.
Two of them -- Rating
and Year are numerical attributes, and we bin them into 3 ranges. For
instance, Rating is classified into three ranges: 1--6,
6--7 and 7--10. Similarly, attribute Year is classified into 3 bins: less
than 1998, 1998--2002 and 2002+. Other attributes include Genre, Maturity
Rating, Director, Actor and Writer. In total, we ended up with 7135
predicates.
In general, we perform the binning in a roughly equi-depth manner, that
is, equal number of (global) followups fall in each bin. The reason for
following such a binning strategy is to remove any prior bias on
selecting predicates.

\begin{table}
	\centering
	\caption {Flixster: Movie (Action) \GB{Predicates/Features}. We collected 7
	attributes. }
	\label{tab:flixster}	
	\begin{scriptsize}
	\begin{tabular}{ | l | l |}    \hline	
		Attribute 	& \# of Predicates/Features \\ \hline \hline			
		Rating 	& 	3 \{1-6, 6-7, 7-10\}	  	\\ \hline
		Year 		& 	3  \{$<$1998, 1998-2002, 2002+\}	\\ \hline
		Genre 		& 	28  \{ex. Comedy\}	\\ \hline		
		Maturity Rating 		& 	94 \{ex. PG-13\} 		\\ \hline
		Director 		& 1239  		\\ \hline
		Actor 		& 	4364 		\\ \hline
		Writer 		& 	1404 		\\ \hline
	\end{tabular}
	\end{scriptsize}
\end{table}

\smallskip
\noindent
\textbf{Twitter.}
Twitter (\url{twitter.com}) is a well known microblogging site where
users post tweets, messages of up to 140 characters, that are broadcast
to users following them. Tweets can be retweeted by receiving users;
this rebroadcasts the tweet to users following the receiver. Thus,
action here is a user posting a tweet (or retweet). 

While collecting data, we focus on tweets that are retweeted, as
retweets are definite indications of flow of influence (or information). 
That is, tweeting is an action and retweeting is evidence of its 
propagation through the network.
Moreover, we restricted our data collection to tweets containing URLs,
as it allows us to compile rich action features, from the webpages
corresponding to the mentioned URLs.

We collected the tweet data using Twitter Streaming API, for 3 weeks
from Tue Jul 24 14:50:07 PDT 2012 to Tue Aug 14 14:57:30 PDT 2012. The
Streaming API permits tracking of specific users, that is, using this
API, we can collect tweets created by these users, and any retweets of
these tweets.  To select these ``source'' users, we exploit the Twitter
Search API. We did not provide any search term in the query, and the API
returned top-20 tweets according to Twitter's internal ranking. Queries
were sent every 5 seconds until a total of 10K source users were
collected. Once we had source users, we collected 
tweet data from \GB{the}
Streaming API, targeting these 10K users. In 3 weeks of data collection,
we accumulated 2.2M (source) tweets from these source users, which were retweeted
92.5M times by 11.8M other users. 

\eat{
\note{Laks: What are the edges? Do they correspond who retweeted whom? Sth else? } 
}

\smallskip
\noindent
\textit{Action Features:}
Tweets may contain user mentions, hashtags and URLs as features.
Usually, URLs are shortened by services like \url{bit.ly}. 
Out of the 2.2M source tweets, 51\%
contained URLs. We focused on these tweets and their retweets. We were
able to expand 98\% of the URLs. To collect features, we 
queried
Delicious (\url{http://delicious.com/}) with URL hostnames and gathered Delicious tags. 
A total of 39K
unique URL tags were found from 28K unique URL hosts. After this
processing, we had 948K source tweets, which received 12.8M retweets. We
consider this sample in our analysis.

\smallskip 
\noindent 
\textbf{Influence Cube and Frequency Distribution of Followups.} From this data, we construct
influence cube $\C$ as described in \textsection\ref{sec:probdef}. That is, the
value in cell $(u,a,v)$ is set to 1 if there is a path from $u$ to $v$ in
the propagation graph corresponding to action $a$. For each user $u$, we then calculate 
the total number of followups, the distribution of which
is presented in Fig.~\ref{fig:dist_followups} (note that both the axes are in log scale). 
As expected, the distribution follows
a power law  in case of Flixster, with 
the power law exponent of  -1.326. 
On the other hand, the distribution in case of Twitter is not exactly power law. This is 
due to the bias in our data collection strategy that favors active users: recall that 
we collected 10K source users by exploiting Twitter Search API. We fit two 
piecewise functions -- the first with exponent -0.462 and the latter with exponent -1.04.

\begin{figure}[t]
\centering
\begin{tabular}{cc}
\hspace{-4mm}\includegraphics[width=4.5cm]{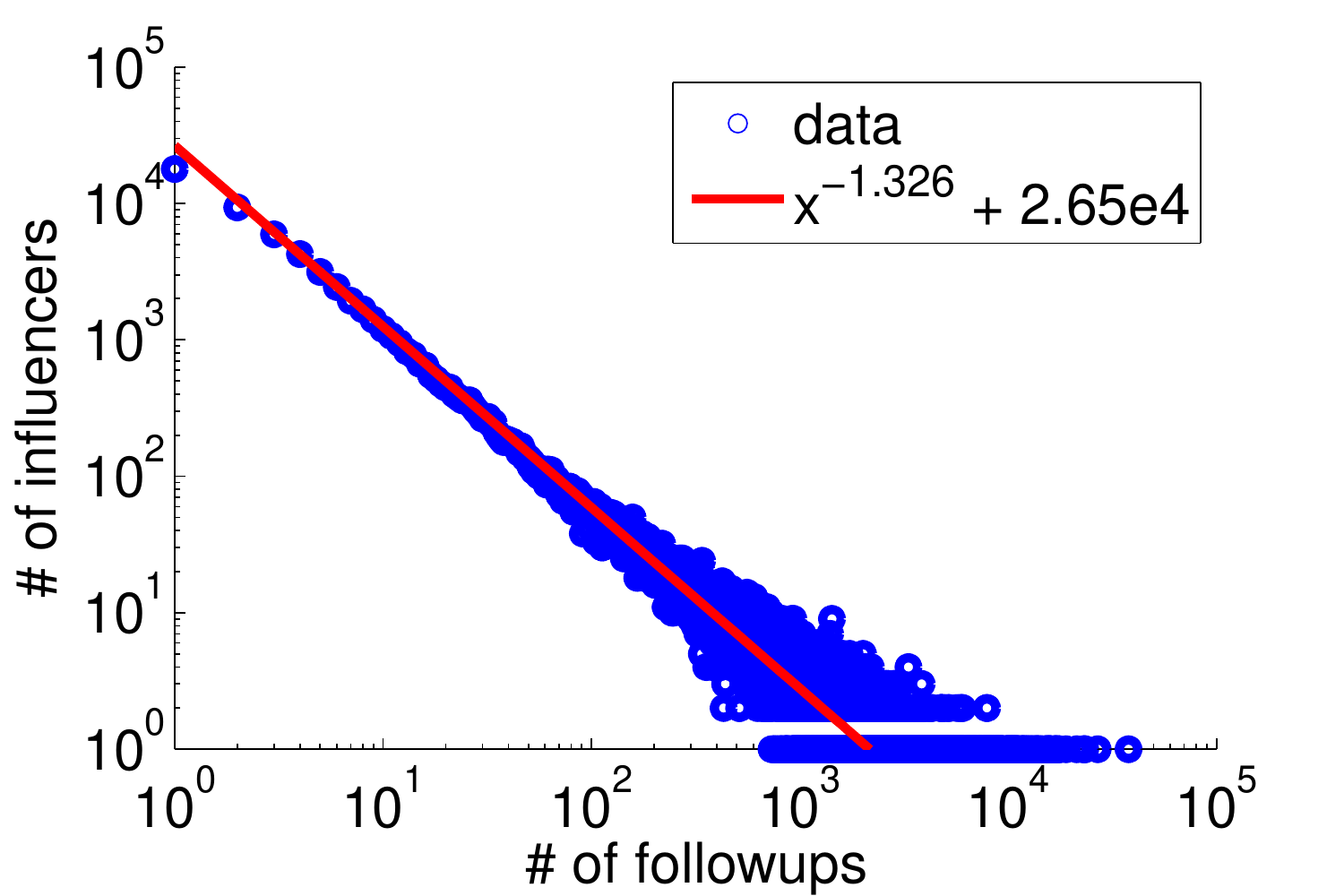}&
\hspace{-4mm}\includegraphics[width=4.5cm]{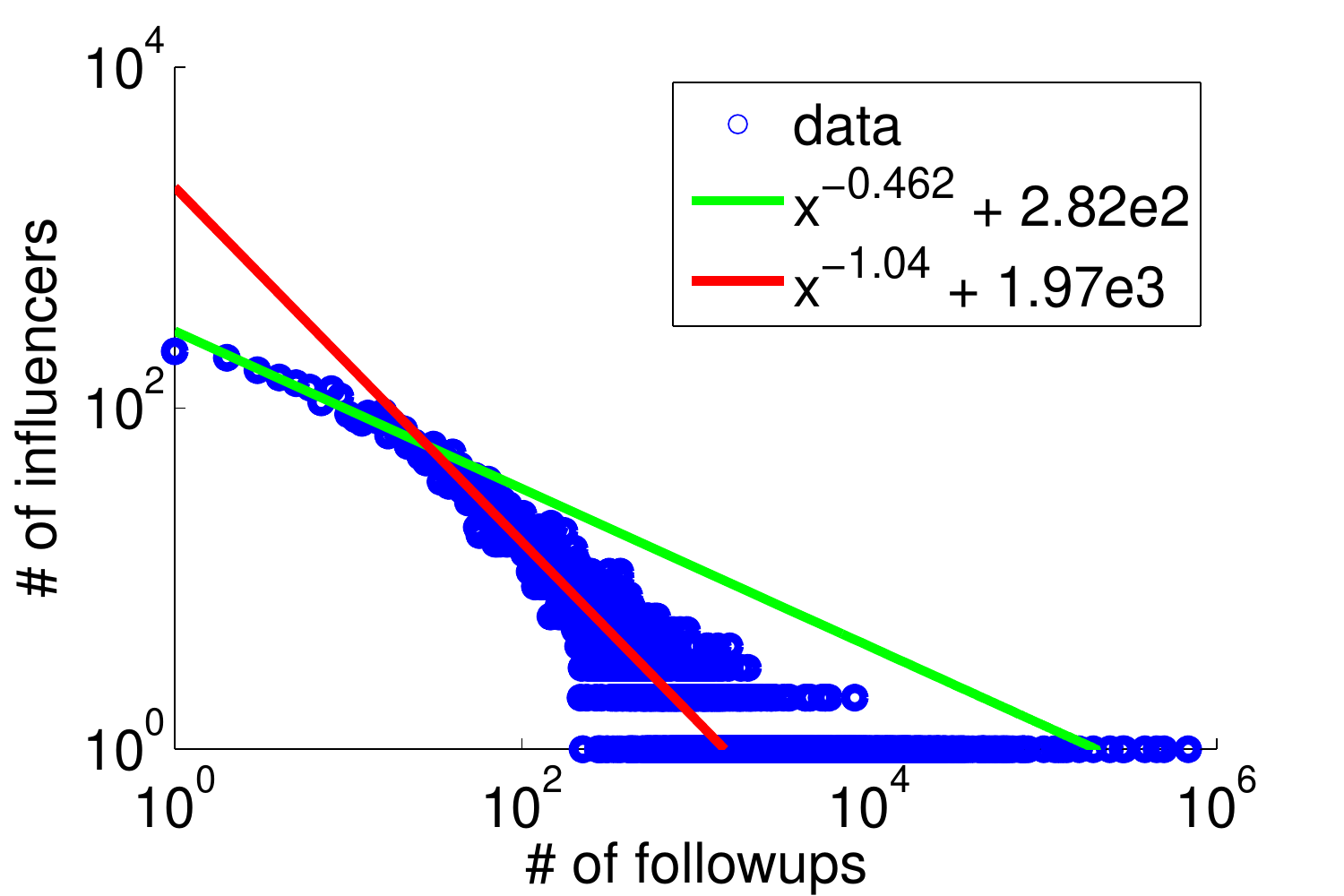}
\end{tabular}
\caption{Frequency distribution of number of followups in Flixster 
(left) and Twitter (right). Both axes are in log scale.\label{fig:dist_followups}} 
\end{figure}

\subsection{Qualitative Analysis}

Through qualitative analysis, we  mainly seek to validate our problem
settings, approach and algorithm. But we note the limitations imposed by
the public datasets available. In Twitter dataset, we know the identity
of the influencers. We can thus validate our approach by checking if the
distribution of influence is along the expected lines. On the other
hand, in Flixster dataset, we do not know the identity of the
influencers, but have the access to both user and action attributes.
Thus, in this case, we can examine the benefits of incorporating
followers' demographics (as we saw in the example shown in Table
\ref{tab:flixster_firstUser}).


\begin{table*}[t]
\parbox{.48\linewidth}{

	\centering
	\caption{Kali: User with 2nd most followups. $k=6$, $k=3$.}
	\begin{scriptsize}
	\begin{tabu}{c c c | c c c }   	
		\multicolumn{3}{c|}{} 	& Actions	& Followers & Followups\\	
		\multicolumn{3}{c|}{}   & (2.8k) & (93)  & (27.0k)\\ \hline
		\multirow{4}{*}{female}	& comedy & pre-1997 & 550 & 62 & 2.8k	\\ \tabucline[2pt]{2-6}
		& thriller & action & 549 & 62 & 3.1k	\\ \tabucline[2pt]{2-6}
		& \multirow{2}{*}{drama} & len:long	& 475 & 62 & 2.8k\\ \tabucline[1pt]{3-6}
		& & len:med& 483 & 62 & 2.1k \\ \tabucline[3pt]{1-6}
		\multirow{2}{*}{male} & thriller &	rated:R & 687 & 28 &  2.6k\\ \tabucline[2pt]{2-6}
		& age:25-34 & pre-1997 & 1.4k & 13 & 2.3k \\ \hline	
		\multicolumn{5}{r}{Total Coverage:} & 51.4\%\\
	\end{tabu}
	
	\label{tab:flixster_secondUser}
	\end{scriptsize}

}
\parbox{.48\linewidth}{
	\centering
	\caption{Julie: User with 3rd most followups. $k=6$, $l=3$.}
	\begin{scriptsize}
	\begin{tabu}{c c c | c c c }   	
		\multicolumn{3}{c|}{} 	&Actions&Followers&Followups\\	
		\multicolumn{3}{c|}{}   & (1.8k) & (73) &(23.5k)\\ \hline
		\multirow{6}{*}{female}	& \multirow{2}{*}{pre-1997} & comedy  & 401 & 63 & 4.3k	\\ \tabucline[1pt]{3-6}
		& & rat:7-10& 317 & 63	& 3.9k\\ \tabucline[2pt]{2-6}
		& \multirow{2}{*}{1998-2002} & comedy	 & 238 & 63 & 3.2k\\ \tabucline[1pt]{3-6}
		& & drama  & 222 & 63& 2.3k \\ \tabucline[2pt]{2-6}
		& rated:R & thriller  & 446 & 63 &  3.7k\\ \tabucline[2pt]{2-6}
		& PG-13 & action  & 228 & 63 & 3.0k\\ \hline	
		\multicolumn{5}{r}{Total Coverage:} & 68.7\%\\
	\end{tabu}
	\label{tab:flixster_thirdUser}
	\end{scriptsize}

}
\end{table*}

\eat{
\begin{table}

	\centering
	\caption{Kali: User with 2nd most followups. $k=6$, $k=3$.}
	\begin{scriptsize}
	\begin{tabu}{c c c | c c c }   	
		\multicolumn{3}{c|}{} 	& Actions	& Followers & Followups\\	
		\multicolumn{3}{c|}{}   & (2.8k) & (93)  & (27.0k)\\ \hline
		\multirow{4}{*}{female}	& comedy & pre-1997 & 550 & 62 & 2.8k	\\ \tabucline[2pt]{2-6}
		& thriller & action & 549 & 62 & 3.1k	\\ \tabucline[2pt]{2-6}
		& \multirow{2}{*}{drama} & len:long	& 475 & 62 & 2.8k\\ \tabucline[1pt]{3-6}
		& & len:med& 483 & 62 & 2.1k \\ \tabucline[3pt]{1-6}
		\multirow{2}{*}{male} & thriller &	rated:R & 687 & 28 &  2.6k\\ \tabucline[2pt]{2-6}
		& age:25-34 & pre-1997 & 1.4k & 13 & 2.3k \\ \hline	
		\multicolumn{5}{r}{Total Coverage:} & 51.4\%\\
	\end{tabu}
	
	\label{tab:flixster_secondUser}
	\end{scriptsize}

	\centering
	\caption{Julie: User with 3rd most followups. $k=6$, $l=3$.}
	\begin{scriptsize}
	\begin{tabu}{c c c | c c c }   	
		\multicolumn{3}{c|}{} 	&Actions&Followers&Followups\\	
		\multicolumn{3}{c|}{}   & (1.8k) & (73) &(23.5k)\\ \hline
		\multirow{6}{*}{female}	& \multirow{2}{*}{pre-1997} & comedy  & 401 & 63 & 4.3k	\\ \tabucline[1pt]{3-6}
		& & rat:7-10& 317 & 63	& 3.9k\\ \tabucline[2pt]{2-6}
		& \multirow{2}{*}{1998-2002} & comedy	 & 238 & 63 & 3.2k\\ \tabucline[1pt]{3-6}
		& & drama  & 222 & 63& 2.3k \\ \tabucline[2pt]{2-6}
		& rated:R & thriller  & 446 & 63 &  3.7k\\ \tabucline[2pt]{2-6}
		& PG-13 & action  & 228 & 63 & 3.0k\\ \hline	
		\multicolumn{5}{r}{Total Coverage:} & 68.7\%\\
	\end{tabu}
	\label{tab:flixster_thirdUser}
	\end{scriptsize}

\end{table}
}

\noindent
\textbf{Formatting the Explanations.}
To avoid clutter, in the explanation tables, we do not mention the name
of each attribute, and instead show its required value directly.  For
instance, consider the example in Table \ref{tab:flixster_firstUser}.
Here, ``rated R'' indicates the feature ``maturity rating = rating R'' and ``thriller''
implies ``genre = thriller''. Similarly, we combine overlapping features
(or predicates) from various explanations, to allow us to visualize the
explanations in a tree structure. As an example, the features ``rated
R'' and ``thriller'' are present in the first two explanations. We order
the features in a manner that minimizes the repetition of the features
in the explanations table.

\noindent
\textbf{Flixster.} 
%
We take the top-3 influencers measured in terms of the number of
followups they recieve,  and apply our algorithm to generate
explanations, for examining their influence distribution. The results
are presented in  Tables \ref{tab:flixster_firstUser},
\ref{tab:flixster_secondUser} and \ref{tab:flixster_thirdUser}. For
simplicity, we refer to the top-3 users as Mike, Kali, and Julie, even
though their identities in Flixster are unknown.  Table
\ref{tab:flixster} shows the complete set of movie features. For users,
we have the features corresponding to attributes gender and age. The
intent should be clear from the context. 

%

While Table \ref{tab:flixster_firstUser} shows the explanations for the
most influential user (please see \textsection\ref{sec:intro} for more
details about this user), 
Tables \ref{tab:flixster_secondUser} and \ref{tab:flixster_thirdUser} show the
explanations for users who received second and third most number of followups, 
Kali and Julie, 
respectively. Kali has rated 2.8K movies and
received 27K followups from her 93 active followers. Julie has rated 1.8K movies
and received 23.5K followups from her 73 active followers. While Kali's
explanations cover 51.4\% of her followups, Julie's explanations cover 
68.7\% of her followups. 

Julie in particular
is influential on female users, with all her followups in the explanations
coming from female users. 
In fact, 63 among
73 of her followers are females. Finally, she is influential on all sorts of
movies (on female users), ranging from comedy, drama, thiller and action. This
implies that {\sl Julie might be a very good seed if the target market is females},
perhaps better than the top-2 users whose influence is distributed among both
males and females. This is the sort of insight that simply cannot be gained by 
viewing network values solely as a scalar!
A final remark about the explanations found is that they are heterogeneous, in that 
they involve a mix of user and action features.

\noindent
\textbf{Twitter.} 
We next analyze the results from Twitter dataset. Recall that we have
user identities of key influencers in Twitter, which allows us to
validate whether the topics on which these influencers reported to be
influential by our algorithm are along the expected lines. This provides
us a nice strategy to validate our problem settings \GB{and approach}.
For instance, we expect  news accounts like New York Times (NYTimes) and
CNN to be influential on  topics like news, politics, media, etc, and
individuals like Tim O'Reilly to be influential on news on software,
tech, programming etc. Fortunately, we were able to generate a rich set
of topics for tweets by expanding their mentioned URLs. 

While we have the identity of the key influencers, for followers, 
we could not collect user
attributes due to demographic data not being available through the API.
Thus, our analysis is restricted to the action
attributes, which consist of tweet topics. We focus on explanations of four
influencers (Twitter accounts) -- New York Times, National Geographic, CNN Breaking News and
Tim O'Reilly, the results of which are presented in Tables \ref{tab:nytimes},
\ref{tab:cnn}, \ref{tab:timreilly} and \ref{tab:nationalgeo}.

\eat{
\begin{table}

	\centering
	\caption{New York Times. $k=6$, $l=4$.}
	\begin{scriptsize}
	\begin{tabu}{c c c c | c c }    
		\multicolumn{4}{c|}{} 	& Followups & Actions	\\ 
		\multicolumn{4}{c|}{} 	&  (246) & (12.5k)\\ \hline	
		\multirow{6}{*}{news}	& \multirow{4}{*}{nytimes} & \multirow{2}{*}{business} & finance & 32	& 745\\ \tabucline[1pt]{4-6}
		& & & media  & 8 & 578	\\ \tabucline[2pt]{3-6}
		& & culture & media	 & 8 & 535 \\ \tabucline[2pt]{3-6}
		& &	journalism & photos  & 14 & 762\\ \tabucline[3pt]{2-6}
		& \multirow{2}{*}{politics} &	media & newspaper	 & 79 & 8.5k\\ \tabucline[2pt]{3-6}
		& & nyt & journalism  & 43 & 4.2k\\ \hline	
		\multicolumn{5}{r}{Total Coverage:} & 82.8\%\\
	\end{tabu}
	\label{tab:nytimes}
	\end{scriptsize}

	\centering
	\caption{CNN breaking news. $k=6$, $l=3$.}
	\begin{scriptsize}
	\begin{tabu}{c c c | c c }    
		\multicolumn{3}{c|}{} 	& Actions	& Followups\\ 
		\multicolumn{3}{c|}{}   & (390) & (56.2k)\\ \hline
		\multirow{6}{*}{news}	& \multirow{4}{*}{cnn} & politics  & 295 & 39.8k	\\ \tabucline[1pt]{3-6}
		& & business & 10 & 1.9k	\\ \tabucline[1pt]{3-6}
		& & breaking-news	 & 75 & 13.3k\\ \tabucline[2pt]{2-6}
		& media & tv  & 288 & 38.5k\\ \tabucline[2pt]{2-6}
		& religion & chistianity	 & 1 & 434\\ \tabucline[2pt]{2-6}
		& magazine & sport & 2 & 265 \\ \hline	
		\multicolumn{4}{r}{Total Coverage:} & 99.8\%\\
	\end{tabu}
	\label{tab:cnn}
	\end{scriptsize}

	\centering
	\caption{Tim O'Reilly. $k=6$, $l=3$.}
	\begin{scriptsize}
	\begin{tabu}{c c c | c c }    
		\multicolumn{3}{c|}{} 	& Actions	& Followups\\ 
		\multicolumn{3}{c|}{} 	& (115) & (3.0k) \\ \hline	
		\multirow{4}{*}{news}	& \multirow{2}{*}{tech} & software & 11	& 410 \\ \tabucline[1pt]{3-6}
		& & business  & 6	& 165\\ \tabucline[2pt]{2-6}
		& media & politics	 & 15 & 513\\ \tabucline[2pt]{2-6}
		& magazine & science & 3 & 113 \\ \tabucline[3pt]{1-6}
		development & programming & opensource	 & 2 & 471\\ \tabucline[3pt]{1-6}
		socialmedia & google & ping.fm  & 5 & 235\\ \hline	
		\multicolumn{4}{r}{Total Coverage:} & 63.4\%\\
	\end{tabu}
	\label{tab:timreilly}
	\end{scriptsize}

	\centering
	\caption{National Geographic. $k=3$, $l=3$.}
	\begin{scriptsize}
	\begin{tabu}{c c c | c c }    
		\multicolumn{3}{c|}{} 	& Actions	& Followups\\
		\multicolumn{3}{c|}{}   & (262) & (23.6k)\\ \hline
		science &	nature & geography  & 116 & 12.8k \\ \tabucline[2pt]{1-5}
		\multirow{2}{*}{travel} &	\multirow{2}{*}{national geographic} & magazine	 & 51 & 4.7k\\ \tabucline[1pt]{3-5}
		& & videos  & 38 & 2.3k\\ \hline	
		\multicolumn{4}{r}{Total Coverage:} & 84.1\%\\
	\end{tabu}
	\label{tab:nationalgeo}
	\end{scriptsize}

\end{table}
}

\begin{table*}[t]
\parbox{.5\linewidth}{
	\centering
	\caption{New York Times. $k=6$, $l=4$.}
	\begin{scriptsize}
	\begin{tabu}{c c c c | c c }    
		\multicolumn{4}{c|}{} 	& Followups & Actions	\\ 
		\multicolumn{4}{c|}{} 	&  (246) & (12.5k)\\ \hline	
		\multirow{6}{*}{news}	& \multirow{4}{*}{nytimes} & \multirow{2}{*}{business} & finance & 32	& 745\\ \tabucline[1pt]{4-6}
		& & & media  & 8 & 578	\\ \tabucline[2pt]{3-6}
		& & culture & media	 & 8 & 535 \\ \tabucline[2pt]{3-6}
		& &	journalism & photos  & 14 & 762\\ \tabucline[3pt]{2-6}
		& \multirow{2}{*}{politics} &	media & newspaper	 & 79 & 8.5k\\ \tabucline[2pt]{3-6}
		& & nyt & journalism  & 43 & 4.2k\\ \hline	
		\multicolumn{5}{r}{Total Coverage:} & 82.8\%\\
	\end{tabu}
	\label{tab:nytimes}
	\end{scriptsize}

}
\parbox{.5\linewidth}{
	\centering
	\caption{CNN breaking news. $k=6$, $l=3$.}
	\begin{scriptsize}
	\begin{tabu}{c c c | c c }    
		\multicolumn{3}{c|}{} 	& Actions	& Followups\\ 
		\multicolumn{3}{c|}{}   & (390) & (56.2k)\\ \hline
		\multirow{6}{*}{news}	& \multirow{4}{*}{cnn} & politics  & 295 & 39.8k	\\ \tabucline[1pt]{3-6}
		& & business & 10 & 1.9k	\\ \tabucline[1pt]{3-6}
		& & breaking-news	 & 75 & 13.3k\\ \tabucline[2pt]{2-6}
		& media & tv  & 288 & 38.5k\\ \tabucline[2pt]{2-6}
		& religion & chistianity	 & 1 & 434\\ \tabucline[2pt]{2-6}
		& magazine & sport & 2 & 265 \\ \hline	
		\multicolumn{4}{r}{Total Coverage:} & 99.8\%\\
	\end{tabu}
	\label{tab:cnn}
	\end{scriptsize}

}
\end{table*}

\begin{table*}[t]
\parbox{.5\linewidth}{
	\centering
	\caption{Tim O'Reilly. $k=6$, $l=3$.}
	\begin{scriptsize}
	\begin{tabu}{c c c | c c }    
		\multicolumn{3}{c|}{} 	& Actions	& Followups\\ 
		\multicolumn{3}{c|}{} 	& (115) & (3.0k) \\ \hline	
		\multirow{4}{*}{news}	& \multirow{2}{*}{tech} & software & 11	& 410 \\ \tabucline[1pt]{3-6}
		& & business  & 6	& 165\\ \tabucline[2pt]{2-6}
		& media & politics	 & 15 & 513\\ \tabucline[2pt]{2-6}
		& magazine & science & 3 & 113 \\ \tabucline[3pt]{1-6}
		development & programming & opensource	 & 2 & 471\\ \tabucline[3pt]{1-6}
		socialmedia & google & ping.fm  & 5 & 235\\ \hline	
		\multicolumn{4}{r}{Total Coverage:} & 63.4\%\\
	\end{tabu}
	\label{tab:timreilly}
	\end{scriptsize}

}
\parbox{.5\linewidth}{
	\centering
	\caption{National Geographic. $k=3$, $l=3$.}
	\begin{scriptsize}
	\begin{tabu}{c c c | c c }    
		\multicolumn{3}{c|}{} 	& Actions	& Followups\\
		\multicolumn{3}{c|}{}   & (262) & (23.6k)\\ \hline
		science &	nature & geography  & 116 & 12.8k \\ \tabucline[2pt]{1-5}
		\multirow{2}{*}{travel} &	\multirow{2}{*}{national geographic} & magazine	 & 51 & 4.7k\\ \tabucline[1pt]{3-5}
		& & videos  & 38 & 2.3k\\ \hline	
		\multicolumn{4}{r}{Total Coverage:} & 84.1\%\\
	\end{tabu}
	\label{tab:nationalgeo}
	\end{scriptsize}
}
\end{table*}

Consider the news accounts NYTimes and CNN first. As we expect, both
these accounts are influential on topics ``news'' and ``politics''. Moreover,
CNN is quite influential on topics like ``tv'' and ``breaking news'', which do
not appear in explanations of NYTimes. This makes sense as CNN is a
television news channel, while NYTimes is a newspaper. Another interesting
observation is that topics like ``religion'' and ``christianity'' appear in CNN
explanations (but not in NYTimes explanations) indicating CNN 
airs programs about religion. In the
sample we collected, CNN tweeted about religion and christianity only once, and
received 434 retweets -- much higher than the average of 56200/390 = 144
retweets per CNN tweet. Similarly, topics ``journalism'' and ``photos'' can be 
found in NYTimes explanations but not in CNN ones, while the topics ``business'' 
and ``politics'' can be found in both. 
Finally, it is interesting to note that these explanations
are able to cover almost all the followups -- 82.8\% for NYTimes and 99.8\% for
CNN, suggesting that these accounts are followed mostly because of their news,
politics, media etc, i.e., the topics represented in the explanations shown in 
these tables.

Next, in Table \ref{tab:timreilly}, we show the explanations of influence of Tim
O'Reilly, the founder of O'Reilly Media and a supporter of the free software and
open source movements. Topics like ``news, tech, media,
software, open source, programming, development, google'' etc.~emerge as the
topics of his influence, which agrees with our expectation. 
Finally, we explore the influence of the National
Geographic Channel in Table \ref{tab:nationalgeo}. This account is influential
on ``science, nature, geography travel'' etc, again consistent with our 
expectations.

Above, we have seen that our algorithm outputs the features (topics) that we
expect these well known accounts to be influential on. These observations
clearly indicate that our problem settings and framework are valid and effective 
for the purpose of digging deep
into the influence spread of influencers and providing explanations.
\emph{When coupled with user features, as in the 
case of Flixster dataset, we are able to answer the questions we raised
in \textsection\ref{sec:intro}:
Where exactly does the influence of an influencer lie? How is it distributed? On
what type of actions is an influencer influential? What are the demographics of
its followers?}

\subsection{Quantitative Analysis}

We next focus on evaluating our algorithm from a quantitative perspective 
and compare our algorithm with other algorithms, in
terms of the coverage achieved (the fraction of total followups), running time, 
and memory usage. 

\noindent
\textbf{Algorithms Compared:} We compare our algorithm, which we refer to as 
\greedy, with the following baselines. 

\textsc{Random}: It selects the features randomly, with probability proportional to 
number of followups covered by each feature. 

\textsc{Most-Popular}: It orders the features by their popularity, i.e.,
number of followups they cover. Then, it picks the top $l$ features 
\GB{that
have yet to be picked 
to build an explanation, this is repeated $k$
times. }
It is an intuitive algorithm, as the features which cover most
followups can be seen as the representative set of features on which the
given influencer is influential. 

\eat{
\note{This doesn't say how the $kl$ features are packed off into 
$k$ explanations. The packaging could be good or bad in term sof coverage 
achieved.} 
}

\exhaustive: It generates one explanation at a time, by exhaustively trying all
possible combinations of features and picking the one that covers the maximum
number of followups (which are not covered by previous explanations). Note that
the number of possible combinations is \GB{${|\calP| \choose l}$}, where $l$ is the
number of features in one explanation. Thus, each explanation is an optimal one, 
i.e., the one with the maximum marginal coverage w.r.t. the previous set of 
explanations chosen. 
Since the objective function \GB{$\cov: 2^{2^{\calP}} \rightarrow \mathbb{R}$}
is monotone non-decreasing and submodular
(see Thm.~\ref{thm:sm1}), 
the set of explanations obtained using this algorithm is an $(1-1/e)$-approximation 
to the optimal solution \cite{submodular}.

Thus, among the algorithms compared, algorithm \exhaustive provides the upper bound on 
the number of
followups that can be possibly covered by the explanations generated. Because of its 
exhaustive nature, we expect the algorithm to be quite slow. 


Unless otherwise stated, on each dataset, we take top-100 influencers
with respect to number of followups they received. The algorithms are
then run on all of 100 influencers, and the median is picked as the
representative value for the comparison. We use median instead of mean,
as it is 
more robust against outliers.
 


\noindent
\textbf{Coverage w.r.t.~change in \textit{k}}: Figure~\ref{fig:cov_k} shows the
variation in relative coverage achieved when $k$ is varied. Recall that $k$
denotes the number of explanations (table rows) generated. Relative
coverage is defined as the fraction of followups that are covered. The parameter
$l$ is fixed to 3 in Flixster and 5 in Twitter. As expected, the (relative)
coverage increases with $k$, but not at the same rate for all algorithms. 
Our algorithm \greedy consistently
performs just as well as
\exhaustive, while beating both \random and \mostpopular by huge margins, on
both the datasets. In fact, the performance of \greedy is almost indistinguishable from that of 
\exhaustive. For instance, on Flixster, with just 6 explanations, \greedy
is able to cover 0.57 fraction of followups, compared to the fraction 0.58 
achieved by \exhaustive. On the other hand, \mostpopular covers 0.25
fraction of followups and \random performs dismally, covering only 0.03
fraction of followups. Moreover, it is worth mentioning that the coverage
achieved quickly saturates (on Flickr only) for both \mostpopular and \random, 
implying that increasing
$k$ would not have helped achieve better coverage from these algorithms.

In case of Twitter, as we can observe, the coverage achieved, is in general higher, with
\greedy covering up to 0.95 fraction of followups, again with 6 explanations. 
Recall, the longer an explanation (higher $l$) the smaller the coverage, in general. 
Despite this, the coverage achieved on Twitter with 6 longer explanations ($l=5$) 
is more than achieved on Flixster with 6 shorter explanations ($l=3$). 
This indicates that the influencers in Twitter are followed due to their niche.  
For
example, news accounts like CNN and New York Times are mostly followed on 
topics like ``news'' and ``politics'' as we saw above. Once again, \greedy and 
\exhaustive significantly outperform other baselines while their performance is 
very close. The relative coverage of \random seems to grow sharply as $k$ increases 
but at $k = 6$ it still performs poorly. The challenge is to cover as much as 
possible with as few but as detailed explanations as possible, and \greedy 
is found to rise to this challenge.

\begin{figure}[t]
\vspace{-2mm} \centering
\begin{tabular}{cc}
\hspace{-4mm}\includegraphics[width=4.5cm]{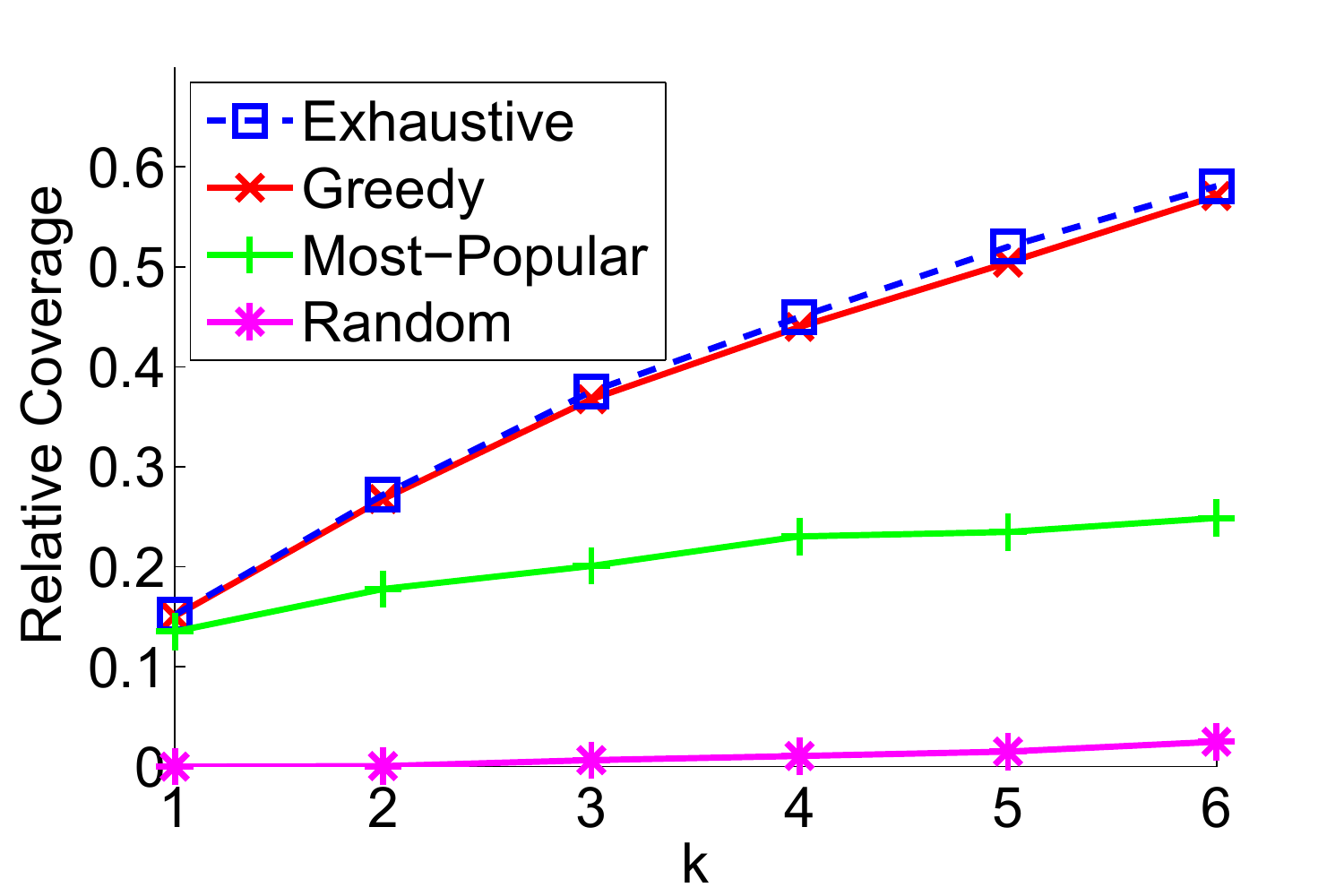}&
\hspace{-4mm}\includegraphics[width=4.5cm]{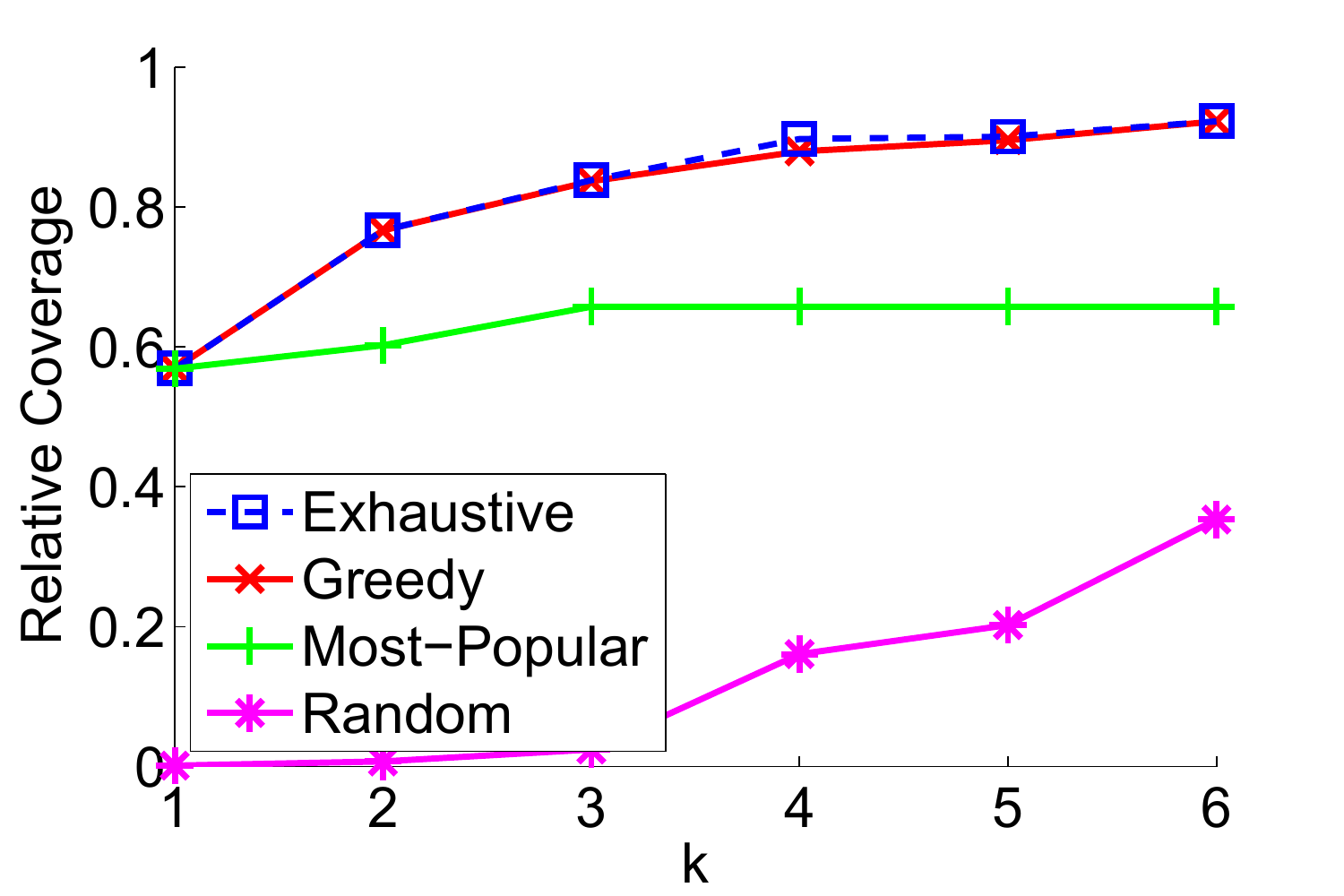}
\end{tabular}
\vspace{-4mm}\caption{Coverage achieved from various algorithms w.r.t.~variation in $k$ in Flixster (left) and Twitter (right). 
$l=3$ in Flixster and $l=5$ in Twitter.\label{fig:cov_k}} \vspace{-2mm}
\end{figure}
\begin{figure}[t]
\vspace{-2mm} \centering
\begin{tabular}{cc}
\hspace{-4mm}\includegraphics[width=4.5cm]{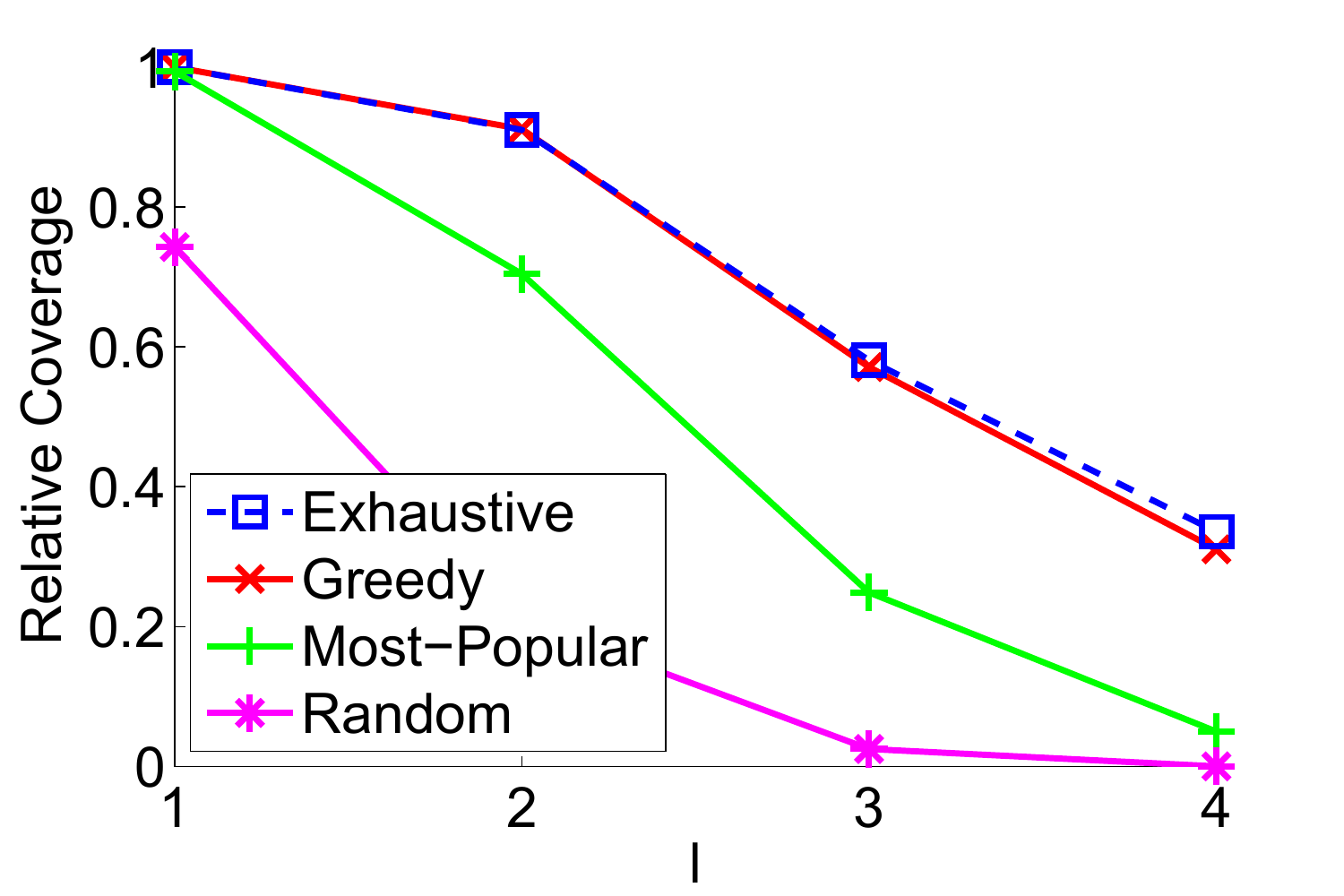}&
\hspace{-4mm}\includegraphics[width=4.5cm]{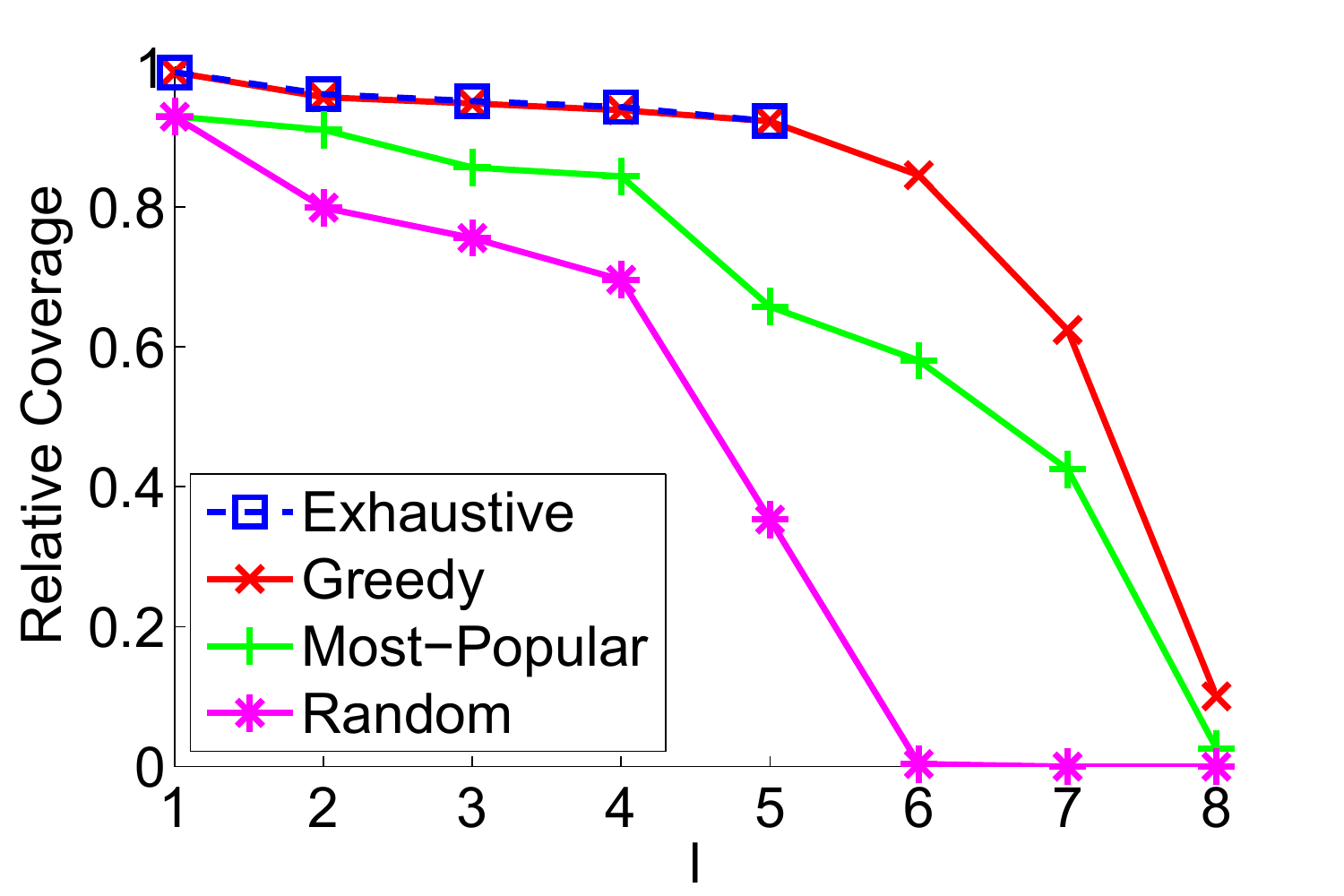}
\end{tabular}
\vspace{-4mm}\caption{Coverage achieved from various algorithms w.r.t.~variation in $l$ in Flixster (left) and Twitter (right). 
$k$ is fixed to 6.\label{fig:cov_l}} \vspace{-2mm}
\end{figure}

\noindent
\textbf{Coverage w.r.t.~change in \textit{l}}:
In Figure \ref{fig:cov_l}, we show the variation in relative coverage
when the parameter $l$, the number of features (table columns) per
explanation, is changed.  As expected,  coverage decreases with the
increase in $l$. Our \greedy algorithm continues to perform quite well.
For instance, on Flixster, it covers 0.31 fraction of followups,
compared to 0.05 and 0.00, the coverage achieved by \mostpopular and
\random, respectively at $l=4$. \exhaustive on the other hand, covers
0.34 fraction of followups.  We see a similar pattern in Twitter
dataset as well. When $l=5$, \greedy covers 0.92 fraction of followups
(same as 0.92 by \exhaustive), while the coverage from \mostpopular
and \random is 0.66 and 0.35. Notice that \exhaustive took too
long to complete for $l > 5$ on Twitter.

\noindent
\textbf{Running Times and Memory Usage.}
Fig.~\ref{fig:running_time_k} 
shows the
running time of various algorithms, on both datasets.
As can be seen, our \greedy algorithm is an order of magnitude faster
than the optimal \exhaustive algorithm. For instance, on Flixster, when
$k=6$ and $l = 3$, while \greedy takes 26 ms to finish, \exhaustive
finishes in 3,748 ms, that is, it takes 144 times longer than \greedy.
The other algorithms -- \mostpopular and \random are faster than \greedy
as we foresaw earlier. They complete in 5 ms and 3 ms.
Similarly, on Twitter, when $k=6$ and $l = 5$, \exhaustive (finishes in 597 ms) is 18 times
slower than \greedy (finishes in 32.5 ms).  On the other hand,
\mostpopular and \random take just 15 ms and 27 ms, respectively.


All the algorithms consume approximately the same amount of memory, up
to a maximum of 446 MB and 2.92 GB on Flixster and Twitter, respectivly.
This is because the memory usage primarily depends on the number of
features and the number of followups they cover. To be precise, \greedy
incurs additional space overhead on account of maintaining heaps of
features, but this additional overhead is under 1 MB, which is
negligible.

\noindent
\emph{In sum, our \greedy algorithm performs essentially as well as \exhaustive
algorithm in terms of coverage achieved, while being much more efficient in
running time.}

\begin{figure}
\vspace{-2mm} \centering
\begin{tabular}{cc}
\hspace{-4mm}\includegraphics[width=4.5cm]{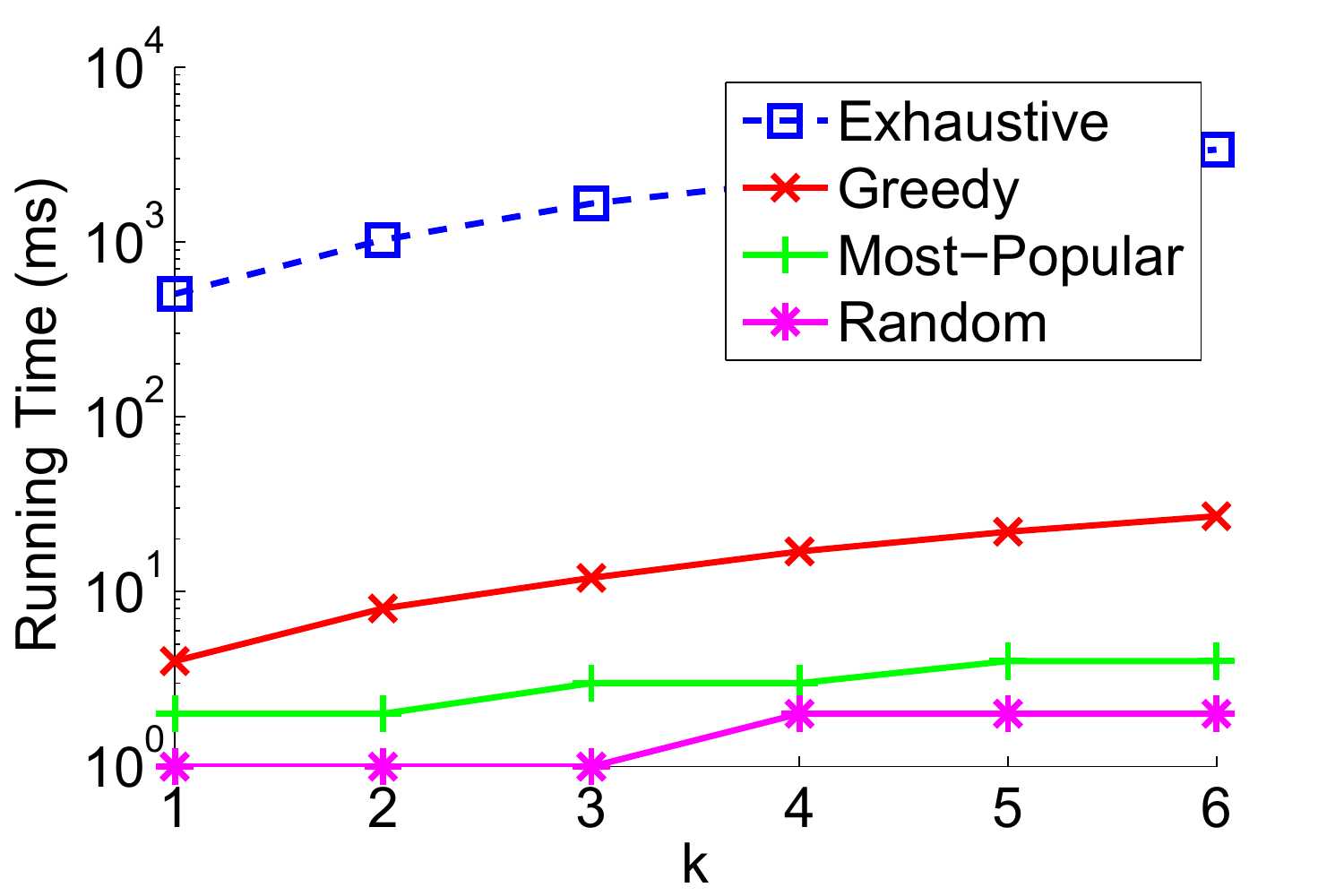}&
\hspace{-4mm}\includegraphics[width=4.5cm]{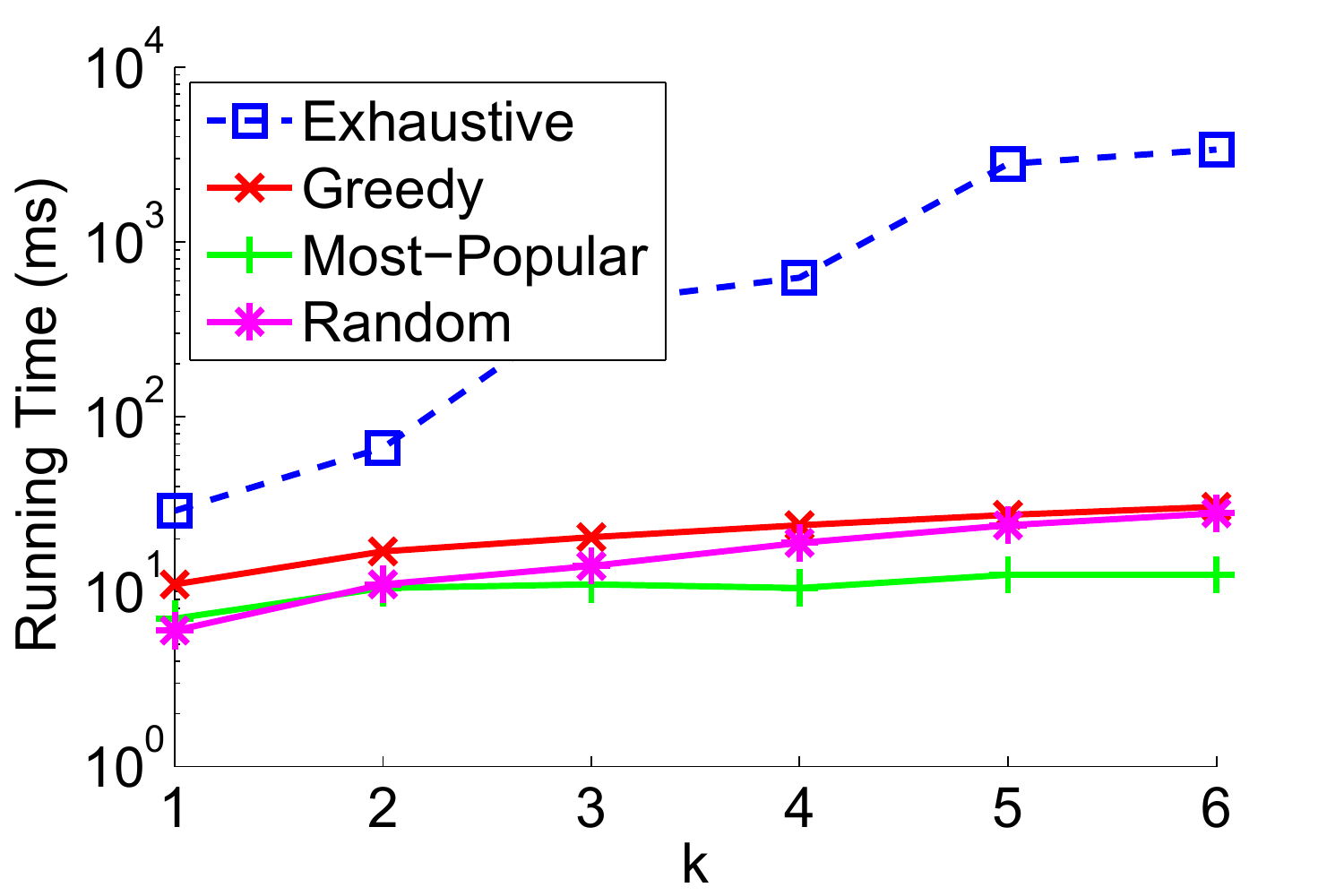}
\end{tabular}
\vspace{-4mm}\caption{Running time comparison for various algorithms on Flixster with $l$ = 3 (left) and Twitter with $l$ = 5 (right). 
\label{fig:running_time_k}} \vspace{-2mm}
\end{figure}







\section{Conclusions and Future Work} 
\label{sec:concl} 

Ever since Domingos and Richardson~\cite{domingos01} introduced the notion of
network value of users in social networks, a lot of work has been done
to identify influencers, community leaders, trendsetters etc. Work in 
this area has been further ignited since Kempe et al. \cite{kempe03} 
popularized influence maximization as a discrete optimization problem.  
Most of
the current approaches largely work as a black box that just outputs a list of
influencers or seeds, along with a scalar which is an estimate of
the expected influence spread. A marketer would 
want to investigate the influence demographics of the seeds returned to her 
and validate them with her own independent survery and/or background knowledge. 
Motivated by this, our goal has been to open up the above 
black box 
and provide informative and crisp
explanations for the influence distribution of influencers, thus allowing 
the marketer to drill down into a seed and address deeper analytic questions 
about what the seed is good for. 

We formalized the above problem as that of finding up to $k$ explanations, each 
containing $l$ or more features, while maximizing the coverage. We showed the problem is not only \NPhard to solve optimally, but is 
\NPhard to approximate within any reasonable factor. Yet, exploiting the nice properties of 
the objective function, we developed a simple greedy algorithm. Our experiments on 
Flixster and Twitter datasets show the validity and usefulness of the 
explanations generated by our framework. Furthermore, they show that the greedy 
algorithm significantly outperforms several natural baselines. 
One of these is an exhaustive approximation algorithm
that by repeatedly finding the explanation with the greatest
marginal coverage gain achieves a (1-1/e)-approximation of the
optimal coverage. However, our
greedy algorithm achieves a coverage very close to that of the
exhaustive approximation algorithm and is an order
of magnitude or more faster. It is interesting to investigate how algorithms for mining maximum frequency item sets of a given cardinality 
(e.g., see \cite{freqitemsetmining}) can be leveraged for finding explanations with the maximum marginal gain. 

Several interesting problems remain open. We give one example. Advertisers 
often like to target users in terms of properties like demographics, instead of 
targeting specific users. E.g., 
if we want to target female college grads in California, what would be an 
effective set of explanations that would describe the influence distribution of this 
demographic? How can we generate these explanations efficiently and validate them? 
\eat{On a more technical level, in this paper, we assumed our input is a social network 
along with an action log. Suppose we are instead given a model consisting of a social network along 
with pairwise user influence probabilities as input. In this case, how should we define 
explanations for {\sl expected} influence spread and how do we generate them efficiently 
and how do we validate them?  } 

\eat{ We have presented an efficient method of assembling an explanation
	of a user's or set of user's influence by identifying feature space
	regions of strong influence. Explaning influence adds value to seed
	set recommendation by giving a means to inspect the results of a
	recommender that otherwise functions as a blackbox. We find that
	proper feature design is critical to achieve balanced use of
	features in explanations, and present techniques to bin features to
	achieve this.  }

\eat{ 
\note{Laks: Need to write from scratch. Main points. Why explain? Why
looking at n/w value as just a scalar is not enough. Main contributions
recalled: f/w, analytical results, algos, experiments, lesson learned,
future work.} 
}


\vspace*{-4pt}
\scriptsize
\bibliographystyle{abbrv} 
\bibliography{singlebib}  
\end{document}